\documentclass[journal,onecolumn,12pt]{IEEEtran}
\usepackage{pslatex}
\usepackage{amsfonts,color,morefloats}
\usepackage{amssymb,amsmath,latexsym,amsthm}
\usepackage{cases}
\usepackage[noadjust]{cite}

\newcommand{\fp}{{\mathbb F}_{p}}
\newcommand{\fq}{{\mathbb F}_{q}}
\newcommand{\ftwo}{{\mathbb F}_{2}}
\newcommand{\fqm}{{\mathbb F}_{q^m}}

\newcommand{\Tr}{{\rm {Tr}}}
\newcommand{\C}{{\mathcal{C}}}

\newcommand{\Supp}{{{\rm Supp}}}

\newtheorem{thm}{Theorem}
\newtheorem{lem}{Lemma}

\newtheorem{example}{Example}
\newtheorem{remark}{Remark}

\begin{document}
\title{Optimal linear codes with few weights from simplicial complexes}
\date{\today}

\author{
Bing Chen, Yunge Xu, Zhao Hu, Nian Li, Xiangyong Zeng
\thanks{B. Chen, Y. Xu and X. Zeng are with the Hubei Key Laboratory of Applied Mathematics, Faculty of Mathematics and Statistics, Hubei University, Wuhan, 430062, China. Z. Hu and N. Li are with the Hubei Provincial Engineering Research Center of Intelligent Connected Vehicle Network Security, School of Cyber Science and Technology, Hubei University, Wuhan, 430062, China. Email: chenbing@hubu.edu.cn, xuy@hubu.edu.cn, zhao.hu@aliyun.com, nian.li@hubu.edu.cn, xiangyongzeng@aliyun.com}
}
\maketitle

\begin{abstract}
Recently, constructions of optimal linear codes from simplicial complexes have attracted much attention and some related nice works were presented. Let $q$ be a prime power. In this paper, by using the simplicial complexes of $\fq^m$ with one single maximal element, we construct four families of linear codes over the ring $\fq+u\fq$ ($u^2=0$), which generalizes the results of [IEEE Trans. Inf. Theory 66(6):3657-3663, 2020]. The parameters and Lee weight distributions of these four families of codes are completely determined. Most notably, via the Gray map, we obtain several classes of optimal linear codes over $\fq$, including (near) Griesmer codes and distance-optimal codes. 
\end{abstract}


\section{Introduction}
Let $\mathbb{F}_{q^m}$ be the finite field with $q^m$ elements and $\mathbb{F}_{q^m}^{*}=\mathbb{F}_{q^m}\backslash\{0\}$, where $q$ is a power of a prime $p$ and $m$ is a positive integer. An $[n, k, d]$ linear code $\mathcal{C}$ over $\fq$ is a $k$-dimensional subspace of $\fq^{n}$ with minimum Hamming distance $d$. Let $A_{i}$ denote the number of codewords with Hamming weight $i$ in a code $\mathcal{C}$ of length $n$. The weight enumerator of $\mathcal{C}$ is defined by
$1+A_{1}z+A_{2}z^{2}+\cdots +A_{n}z^{n}$. The sequence $(1, A_{1}, A_{2}, \cdots ,A_{n})$ is called the weight distribution of $\mathcal{C}$.
A code is said to be a $t$-weight code if the number of nonzero $A_{i}$ in the sequence $(A_{1}, A_{2}, \cdots ,A_{n})$ is equal to $t$. Linear codes with few weights have applications in secret sharing schemes \cite{ARJD,CCDY}, authentication codes \cite{DCHT,DW}, association schemes \cite{CAGJ}, strongly regular graphs and some other fields.

An $[n,k,d]$ linear code $\mathcal{C}$ over $\fq$ is said to be distance-optimal if no $[n,k,d+1]$ code exists (i.e., this code has the largest minimum distance for given length $n$ and dimension $k$) and it is called almost distance-optimal if there exists an $[n,k,d+1]$ distance-optimal code. An $[n,k,d]$ linear code $\mathcal{C}$ is called optimal (resp. almost optimal) if its parameters $n$, $k$ and $d$ (resp. $d+1$) meet a bound on linear codes with equality \cite{HWPV}. The Griesmer bound \cite{JHG,GSJS} for an $[n,k,d]$ linear code $\C$ over $\fq$ is given by
\[ n\geq  g(k,d):=\sum_{i=0}^{k-1} \lceil \frac{d}{q^i}\rceil,\]
where $\lceil \cdot \rceil$ denotes the ceiling function. An $[n,k,d]$ linear code $\mathcal{C}$ is called a Griesmer code (resp. near Griesmer code) if its parameters $n$ (resp. $n-1$), $k$ and $d$ achieve the Griesmer bound.

In 2007, Ding and Niederreiter \cite{DN} introduced a nice and generic way to construct linear codes via trace functions. Let $D \subset \fqm$ and define
\begin{equation*} \label{CD}
\C_D=\{c_{a}=(\Tr_{q}^{q^{m}}(ax))_{x\in D}: a\in \fqm\},
\end{equation*}
where $\Tr_{q}^{q^{m}}(\cdot)$ is the trace function from $\fqm$ to $\fq$. Then $\C_{D}$ is a linear code of length $n:=|D|$ over $\mathbb{F}_{q}$ and the set $D$ is called the defining set of $\C_D$. Many attempts have been made in this direction by selecting proper defining set to obtain good or optimal linear codes, see, for example, \cite{DING,HZHZ,HDWZ,JYHLL,ZZNL,MQRT,MAS,LXYQ} and references therein.

Let $R$ be a finite commutative ring and $R_{m}$ be an extension of $R$ of degree $m$. A trace code over $R$ with defining set $L \subset R_{m}^{*}$ is defined by
\begin{equation} \label{CL}
\C_L=\{c_{a}=(\Tr(ax))_{x\in L}: a\in R_{m}\}
\end{equation}
where $R_{m}^{*}$ is the multiplicative group of units of $R_{m}$ and $\Tr(\cdot)$ is a linear function from $R_{m}$ to $R$.
Using the construction above, some good linear codes over rings were contructed in the previous works, see, for example, \cite{MSYLPS,SMWRLY,MSYGPS,HLYM,YWXZQY,HCLZ} and references therein.

Recently, constructing optimal or good linear codes from simplicial complexes has attracted much attention from researchers. Some optimal linear codes over the finite field $\fq$ or the ring $R$ has been constructed in this direction. For the results on constructing linear codes over finite fields from simplicial complexes, we refer the readers to \cite{CHJY,JYHLL,HXLZWT} and references therein. As for linear codes over the ring $R$, to the best of my knowledge, Wu et al. (in 2020) were the first to construct linear codes over $R=\ftwo+u\ftwo$ ($u^2=0$) from simplicial complexes of $\ftwo^m$ with one maximal element, from which two classes of optimal few-weight binary codes were obtained via the Gray map. Later, Wu et al. \cite{WYHJ} constructed linear codes over $\fp+u\fp$ ($u^2=0$ and $p$ is a prime) from simplcial complexes of $\fp^m$ with one maximal element, and obtained $2(p-1)$ classes of $p$-ary distance-optimal linear codes. In 2021, Li et al. \cite{LXSM} constructed a family of linear codes over $R=\ftwo+u\ftwo+u^2\ftwo$ ($u^3=0$) from simplicial complexes of $\ftwo^m$ with one maximal element, from which a new family of optimal binary few-weight codes was derived. In the same year, Shi et al. \cite{SMLX} constructed two classes of linear codes over $\fp+u\fp$ ($u^2=u$) from simplcial complexes of $\fp^m$ with one maximal element, and two classes of distance-optimal $p$-ary linear codes were presented by using the Gray map. In 2024, Wu et al. \cite{WLZX} prsented interesting constructions of linear codes over ${\mathbb Z}_{4}$ by using the simplicial complexes of $\ftwo^m$ with one or two maximal elements.

In this paper, inspired by the previous works, we investigate the linear codes over the ring $R=\fq+u\fq$ ($u^2=0$) by employing the simplicial complexes of $\fq^m$. Let $\Delta_{A}$, $\Delta_{B}$ and $\Delta_{B'}$ be simplicial complexes of $\fq^m$ with one maximal element, where $A$ and $B$ are subsets of $[m]$ and $B'$ is a subset of $B$. We focus on the four families of linear codes $C_{L}$ over $R=\fq+u\fq$ defined by \eqref{CL}
with the following defining sets:
\begin{align} \label{defing-set-L}
  &1):\,\,L=\Delta_{A}+u(\Delta_{B}\setminus \Delta_{B'});   \nonumber\\
  &2):\,\,L=\Delta_{A}^c+u(\Delta_{B}\setminus \Delta_{B'});  \nonumber  \\
  &3):\,\,L=\Delta_{A}+u(\Delta_{B}\setminus \Delta_{B'})^c;  \nonumber\\
  &4):\,\,L=\Delta_{A}^c+u(\Delta_{B}\setminus \Delta_{B'})^c, 
\end{align}
where $\Delta^c$ denotes the complement of $\Delta$ for a set $\Delta$ of $\fqm$.
Notice that the codes $\C_{L}$ with defining sets $1)$ and $2)$ are reduced to the codes studied in \cite{YWXZQY} if $q=2$ ane $|B|=m$. By employing some detailed calculations on certain exponential sums, we obtained infinite families of few-weight linear codes with flexible parameters, and completely determined the parameters and Lee weight distributions of these four families of linear codes over $\fq+u\fq$. Most notably, under the Gray map $\phi $, several classes of optimal linear codes over $\fq$ are derived by characterizing the optimality of the Gray images $\phi(\C_{L})$ using the Griesmer bound, which include the (near) Griesmer codes and distance-optimal linear codes.

The rest of this paper is organized as follows. In Section \ref{sect-2}, we introduce some concepts and results. In Sections \ref{sect-3}, we determine the parameters and Lee weight distributions of these four infinite families of linear codes $\C_{L}$ over $\fq+u\fq$ with defining sets given by \eqref{defing-set-L}. In Section \ref{sect-4}, by employing the Gray map, we obtain several infinite families of optimal linear codes over $\fq$ and present some examples. Section \ref{sect-conclusion} concludes this paper.

\section{Preliminaries} \label{sect-2}
In this section, we introduce some notation, definitions and lemmas which will be used later.
Starting from now on, we adopt the following notation unless otherwise stated:
\begin{itemize}
  \item $q$ is a power of a prime $p$, and $m$ is a positive integer.
  \item Identify the vector space $\fq^m$ with the finite field $\fqm$ since  $\fq^m$ is isomorphic to $\fqm$.
  \item Let $R=\fq+u\fq$ with $u^2=0$.
  \item For a positive integer $e$, $[e]=\{1,\cdots,e\}$.
  \item For a set $S$, $|S|$ denotes the cardinality of $S$.
  \item For a vector $x\in \fq^m$, $wt(x)$ denotes the Hamming weight of $x$; for a vector $y \in R^n$, $wt_{L}(y)$ denotes the Lee weight of $y$.
  \item $\Tr_{q}^{q^{m}}(\cdot)$ denotes the trace function from $\fqm$ to $\fq$.
  \item $\chi(\cdot)$ denotes the canonical additive character of $\fq$, i.e. $\chi(\cdot)=\zeta_{p}^{\Tr_{p}^{q}(\cdot)}$, where $\zeta_{p}$ is a primitive complex $p$-th root of unity and $\Tr_{p}^{q}(\cdot)$ is the trace function from $\fq$ to $\fp$.
  \item $\Delta_{A}$ is a simplicial complex of $\fq^m$ with one maximal element, where $A\subseteq [m]$ is the support of the maximal element.
\end{itemize}

\subsection{The definition of simplicial complexes of $\fq^m$}

For two vectors $u=(u_{1},u_{2},\ldots,u_{m})$ and $v=(v_{1},v_{2},\ldots,v_{m})$ in $\fq^m$, we say that $u$ covers $v$, denoted $v\preceq u$, if $\Supp(v)\subseteq \Supp(u)$, where $\Supp(u)= \{1 \leq i \leq m : u_{i}\ne 0\}$ is the support of $u$. A subset $\Delta$ of $\fq^m$ is called a \textbf{simplicial complex} if $u\in \Delta$ and $v\preceq u$ imply $v\in\Delta$.

For a simplicial complex $\Delta \subseteq \fq^m$, an element $u$ in $\Delta$ with entries $0$ or $1$ is said to be maximal if there is no element $v\in \Delta$ such that $\Supp(u)$ is a proper subset of $\Supp(v)$. Let $\mathcal{F}=\{F_{1},F_{2},\ldots,F_{h}\}$ be the set of maximal elements of $\Delta$, where $h$ is the number of maximal elements in $\Delta$ and $F_{i}$'s are maximal elements of $\Delta$. Let $A_{i}=\Supp(F_{i})$ for $1\leq i \leq h$, which implies $A_{i}\subseteq [m]$. Let $\mathcal{A}=\{A_{1},A_{2},\ldots,A_{h}\}$ be the set of supports of maximal elements of $\Delta$, and $\mathcal{A}$ is said to be the support of $\Delta$, denoted $\Supp(\Delta)=\mathcal{A}$. Then one can see that a simplicial complex $\Delta$ is uniquely generated by $\mathcal{A}$, denoted $\Delta=\langle \mathcal{A}\rangle$.
Notice that both the set of maximal elements $\mathcal{F}$ and the support $\mathcal{A}$ of $\Delta$ are unique for a fixed simplicial complex $\Delta$.

It should be noted that the simplicial complex $\Delta_{A}=\langle \{A\} \rangle$ with exactly one maximal element is an $|A|$-dimensional subspace of $\fq^m$. Let $\{\alpha_{1},\ldots,\alpha_{m}\}$ be a basis of $\fqm$ over $\fq$. Then $\Delta_{A}$ can be viewed as an $|A|$-dimensional $\fq$-subspace of $\fqm$ spanned by the set $\{\alpha_{i}: i \in A\}$. Specially, when $A=\emptyset$ (i.e., $|A|=0$), we have $\Delta_{A}=\{0\}$. For more details on simplicial complexes, we refer the readers to \cite{CHJY,HXLZWT,YWXZQY}.

\subsection{Linear codes over the ring $R$}

Let $R=\fq+u\fq$ with $u^2=0$. A linear code $\C$ of length $n$ over $R$ is an $R$-submodule of $R^{n}$. For any $a+ub\in R$ where $a,b\in \fq$, the Gray map $\phi$ from $R$ to $\fq^2$ is defined by
\[\phi : R \rightarrow \fq^2, a+ub\mapsto (b,a+b).\]
Any vector $\mathbf{x}\in R^{n}$ can be written as $\mathbf{x}=\mathbf{a}+u\mathbf{b}$ where $\mathbf{a},\mathbf{b}\in\fq^{n}$. The map $\phi$ is a bijection, which can be extended naturally from $R^{n}$ to $\fq^{2n}$ as follows:
\[\phi : R^{n} \rightarrow \fq^{2n}, \mathbf{x}=\mathbf{a}+u\mathbf{b}\mapsto (\mathbf{b},\mathbf{a}+\mathbf{b}).\]
The Hamming weight $wt(\mathbf{a})$ of a vector $\mathbf{a}\in \fq^{n}$ is the number of nonzero coordinates in $\mathbf{a}$. The Lee weight $wt_{L}(\mathbf{a}+u\mathbf{b})$ of a vector $\mathbf{a}+u\mathbf{b} \in R^n$ is the Hamming weight of its Gray image $\phi(\mathbf{a}+u\mathbf{b})$ as follows:
\begin{equation*} \label{eq-wt}
wt_{L}(\mathbf{a}+u\mathbf{b})=wt(\mathbf{b})+wt(\mathbf{a}+\mathbf{b}).
\end{equation*}
The Lee distance of $\mathbf{x},\mathbf{y}\in R^{n}$ is defined as $wt_{L}(\mathbf{x}-\mathbf{y})$. It is easy to check that the Gray map is an isometry from $(R^{n},d_{L})$ and $(\fq^{2n},d_{H})$.

Let $\mathcal{R}=\fqm+u\fqm$ with $u^2=0$. Let $F$ be the Frobenius operator over $\mathcal{R}$ defined by $F(a+ub)=a^q+ub^q$. The trace function $\Tr(\cdot)$ is defined by
\[\Tr=\sum_{i=0}^{m-1}F^{i}: \mathcal{R} \rightarrow R, a+ub\mapsto \sum_{i=0}^{m-1}F^{i}(a+ub)=\sum_{i=0}^{m-1}(a^{q^i}+ub^{q^i}).\]
By the definition above, it can be readily verified that
\begin{equation*} \label{eq-Tr}
\Tr(a+ub)=\Tr_{q}^{q^m}(a)+u\Tr_{q}^{q^m}(b),
\end{equation*}
where $\Tr_{q}^{q^{m}}(\cdot)$ denotes the trace function from $\fqm$ to $\fq$.

With the discussion above, the Lee weight of the trace code $\C_{L}$ for a general defining set $L$ can be determined by the following lemma.

\begin{lem} (\cite{HCLZ}) \label{lem-wtL}
Let $L=L_{1}+uL_{2}$ where $L_{1}, L_{2}\in \fqm$. Then $\C_{L}$ is a code of length $|L|$ over $R$ and for any $a+ub \in \mathcal{R}\backslash \{0\}$, the Lee weight of the codeword $c_{a+ub}$ in $\C_{L}$ is $wt_{L}(c_{a+ub})=2|L|-\Omega$ where
\[\Omega=\frac{1}{q}\sum_{u\in \fq}\sum_{y \in L_{2}}\chi(u\Tr_{q}^{q^m}(ay))\sum_{x\in L_{1}}(\chi(u\Tr_{q}^{q^m}(bx))+\chi(u\Tr_{q}^{q^m}((a+b)x))).\]
\end{lem}

\subsection{Useful auxiliary results}
For an $r$-dimensional $\fq$-subspace $H$ of $\fqm$, the dual of $H$ is defined by
\[H^{\bot}=\{v\in \fqm: \Tr_{q}^{q^m}(uv)=0 \mbox{ for all } u\in H\}.\]
The dual $H^{\bot}$ is an $(m-r)$-dimensional $\fq$-subspace of $\fqm$. Let $\{\alpha_{1},\ldots,\alpha_{m}\}$ be a basis of $\fqm$ over $\fq$ and $\{\beta_{1},\ldots,\beta_{m}\}$ be its dual basis. For the simplicial complex $\Delta_{A}$, its dual $\Delta_{A}^{\bot}$ is the $m-|A|$-dimensional $\fq$-subspace of $\fqm$ spanned by the set $\{\beta_{j}: j \in [m]\setminus A\}$.

By using the relevant result in \cite{YLFL}, we can give the following lemma regarding the exponential sum on the simplicial complexex $\Delta_{A}$ of $\fq^m$ with exactly one maximal element, which will be helpful to prove our main results.

\begin{lem} \label{lem-expsum-H}
Let $\Delta_{A}=\langle \{A\} \rangle$ be the simplicial complex of $\fqm$ with exactly one maximal element. Then for $y\in \fqm^*$ we have
\begin{eqnarray*}
\sum_{x\in \Delta_{A}}\zeta_{p}^{\Tr_{p}^{q^m}(yx)}=\left\{\begin{array}{ll}
q^{|A|},    &   \mbox{ if } y\in \Delta_{A}^{\bot};\\
0,   &   \mbox{ otherwise}.
\end{array} \right.
\end{eqnarray*}
\end{lem}


\section{Four families of linear codes over $\fq+u\fq$} \label{sect-3}
In this section, we construct four families of linear codes over $R=\fq+u\fq$ with $u^2=0$ by employing the simplicial complexes of $\fq^m$ with one maximal element. With detailed computation on some exponential sums, the parameters and the Lee weight distributions of these codes are completely determined.

\subsection{The first class of linear codes $\C_{L}$ with $L=\Delta_{A}+u(\Delta_{B}\setminus \Delta_{B'})$}
\begin{thm} \label{code-1}
Let $m$ be a positive integer. Let $\Delta_{A}$, $\Delta_{B}$ and $\Delta_{B'}$ be simplicial complexes of $\fqm$, where $A\subseteq [m]$ and $B' \subset B\subseteq [m]$. Assume that $|A|+|B'|> 0$. Denote $L=\Delta_{A}+u(\Delta_{B}\setminus \Delta_{B'})$. Then $\C_{L}$ defined by \eqref{CL} is a $4$-weight code of length $q^{|A|}(q^{|B|}-q^{|B'|})$, size $q^{|A|+|A\cup B|}$, and its Lee weight distribution is given by
\begin{center}
\begin{tabular}{llll}
\hline Weight $w$ & Multiplicity $A_{w}$\\ \hline
$0$ & $1$ \\
$2(q-1)q^{|A|+|B|-1}$               & $q^{|A\cup B|-|A\cup B'|}-1$ \\
$2(q-1)q^{|A|-1}(q^{|B|}-q^{|B'|})$ & $q^{|A|+|A\cup B|}-2q^{|A\cup B|-|B'|}+q^{|A\cup B|-|A\cup B'|}$ \\ 
$(q-1)q^{|A|-1}(q^{|B|}-q^{|B'|})$ & $2(q^{|A\cup B|-|B|}-1)$ \\
$(q-1)q^{|A|-1}(2q^{|B|}-q^{|B'|})$ & $2(q^{|A\cup B|-|B'|}-q^{|A\cup B|-|B|}-q^{|A\cup B|-|A\cup B'|}+1)$ \\
\hline
\end{tabular}
\end{center}
\end{thm}

\begin{proof}
It is easy to check that the length of $\C_{L}$ is $n:=q^{|A|}(q^{|B|}-q^{|B'|})$. By Lemma \ref{lem-wtL}, for $a+ub \in \mathcal{R}\backslash \{0\}$, the Lee weight of the codeword $c_{a+ub}$ in $\C_{L}$ is
\[wt_{L}(c_{a+ub})=2q^{|A|}(q^{|B|}-q^{|B'|})-\Omega,\]
where
\[\Omega=\frac{1}{q}\sum_{u\in \fq}\sum_{y \in \Delta_{B}\setminus \Delta_{B'}}\chi(u\Tr_{q}^{q^m}(ay))\sum_{x\in \Delta_{A}}(\chi(u\Tr_{q}^{q^m}(bx))+\chi(u\Tr_{q}^{q^m}((a+b)x))).\]

By Lemma \ref{lem-expsum-H}, for $u\in \fq^*$, one can obtain that
\begin{align} \label{eq-thm2-Delta}
   \sum_{y\in \Delta_{B}\setminus \Delta_{B'}}\chi(u\Tr_{q}^{q^m}(ay))=&\sum_{y\in \Delta_{B}}\chi(u\Tr_{q}^{q^m}(ay))-\sum_{y\in \Delta_{B'}}\chi(u\Tr_{q}^{q^m}(ay))\nonumber\\
   =&\left\{\begin{array}{ll}
      q^{|B|}-q^{|B'|},    &   \mbox{ if } a\in \Delta_{B}^{\bot};\\
      -q^{|B'|},    &   \mbox{ if } a\notin \Delta_{B}^{\bot}, a\in \Delta_{B'}^{\bot};\\
      0,    &   \mbox{ if } a\notin \Delta_{B'}^{\bot},
      \end{array} \right.
\end{align}
where $\Delta_{B}^{\bot}\subset \Delta_{B'}^{\bot}$. 
To determine the value of $\Omega$, we consider the following three cases.

Case $(1)$: $a\in \Delta_{B}^{\bot}$. Then by Lemma \ref{lem-expsum-H} we have
\begin{align*}
\Omega=&\frac{2}{q}n+\frac{1}{q}(q^{|B|}-q^{|B'|})\sum_{u\in \fq^*}\sum_{x\in \Delta_{A}}(\chi(u\Tr_{q}^{q^m}(bx))+\chi(u\Tr_{q}^{q^m}((a+b)x))) \\
=&\left\{\begin{array}{ll}
2q^{|A|}(q^{|B|}-q^{|B'|}), & \mbox{if } b\in \Delta_{A}^{\bot}, a+b \in \Delta_{A}^{\bot}; \\
2q^{|A|-1}(q^{|B|}-q^{|B'|}),  & \mbox{if } b\notin \Delta_{A}^{\bot}, a+b\notin \Delta_{A}^{\bot};\\
(q+1)q^{|A|-1}(q^{|B|}-q^{|B'|}),  & \mbox{otherwise}.
\end{array}\right.
\end{align*}
Thus, for $a\in \Delta_{B}^{\bot}$, one gets
\[ wt_{L}(c_{a+ub})=\left\{\begin{array}{ll}
0, & \mbox{if } b\in \Delta_{A}^{\bot}, a+b \in \Delta_{A}^{\bot}; \\
2(q-1)q^{|A|-1}(q^{|B|}-q^{|B'|}),  & \mbox{if } b\notin \Delta_{A}^{\bot}, a+b\notin \Delta_{A}^{\bot};\\
(q-1)q^{|A|-1}(q^{|B|}-q^{|B'|}),  & \mbox{otherwise}.
\end{array}\right.\]

Case $(2)$: $a\notin \Delta_{B}^{\bot}$ and $a\in \Delta_{B'}^{\bot}$. Then we have
\begin{align*}
\Omega=&\frac{2}{q}n-q^{|B'|-1}\sum_{u\in \fq^{*}}\sum_{x\in \Delta_{A}}(\chi(u\Tr_{q}^{q^m}(bx))+\chi(u\Tr_{q}^{q^m}((a+b)x)))\\
=&\left\{\begin{array}{ll}
2q^{|A|-1}(q^{|B|}-q^{|B'|+1}), & \mbox{if } b\in \Delta_{A}^{\bot}, a+b \in \Delta_{A}^{\bot};  \\
2q^{|A|-1}(q^{|B|}-q^{|B'|}),  & \mbox{if } b\notin \Delta_{A}^{\bot}, a+b \notin \Delta_{A}^{\bot}; \\
q^{|A|-1}(2q^{|B|}-(q+1)q^{|B'|}),  & \mbox{otherwise}.
\end{array}\right.
\end{align*}
Thus, for $a\notin \Delta_{B}^{\bot}$ and $a\in \Delta_{B'}^{\bot}$, it gives
\begin{equation}\label{wt-case2}
wt_{L}(c_{a+ub})=\left\{\begin{array}{ll}
2(q-1)q^{|A|+|B|-1}, & \mbox{if } b\in \Delta_{A}^{\bot}, a+b \in \Delta_{A}^{\bot}; \\
2(q-1)q^{|A|-1}(q^{|B|}-q^{|B'|}),  & \mbox{if } b\notin \Delta_{A}^{\bot}, a+b \notin \Delta_{A}^{\bot}; \\
(q-1)q^{|A|-1}(2q^{|B|}-q^{|B'|}),  & \mbox{otherwise}.
\end{array}\right.
\end{equation}

Case $(3)$: $a\notin \Delta_{B'}^{\bot}$. Then we have $\Omega=\frac{2}{q}n$, which indicates
\[ wt_{L}(c_{a+ub}) =2(q-1)q^{|A|-1}(q^{|B|}-q^{|B'|}).\]

With the discussion above, $wt_{L}(c_{a+ub})=0$ if and only if $a\in \Delta_{B}^{\bot}$, $b\in \Delta_{A}^{\bot}$ and $a+b \in \Delta_{A}^{\bot}$. This indicates 
\begin{align*}
  A_{0}=&|\{(a,b)\in \fqm^{2}: a\in \Delta_{B}^{\bot}, b\in \Delta_{A}^{\bot}, a+b \in \Delta_{A}^{\bot}\}|\\
  =&|\{(a,b)\in \fqm^{2}: a\in \Delta_{B}^{\bot}, a \in \Delta_{A}^{\bot}, b\in \Delta_{A}^{\bot}\}|\\
  =&q^{m-|A|}|\{a\in \fqm: a\in \Delta_{B}^{\bot}, a \in \Delta_{A}^{\bot}\}|\\
  =&q^{m-|A|}|\Delta_{B}^{\bot}\cap \Delta_{A}^{\bot}|\\
  =&q^{2m-|A|-|A\cup B|}.
\end{align*} 
Here the second equality holds due to the fact that the addition operation in $\Delta_{A}^{\bot}$ is closed, and the last equality holds since $|\Delta_{B}^{\bot} \cap \Delta_{A}^{\bot}|=|\Delta_{A\cup B}^{\bot}|=q^{m-|A\cup B|}$.
This means that the size of $\C_L$ is $q^{|A|+|A\cup B|}$.
Moreover, $\C_L$ has the following four nonzero Lee weights:
$w_{1}=2(q-1)q^{|A|+|B|-1}$, $w_{2}=2(q-1)q^{|A|-1}(q^{|B|}-q^{|B'|})$, $w_{3}=(q-1)q^{|A|-1}(q^{|B|}-q^{|B'|})$ and $w_{4}=(q-1)q^{|A|-1}(2q^{|B|}-q^{|B'|})$.

It then follows from \eqref{wt-case2} that
\begin{align*}
A_{w_{1}}=&|\{(a,b)\in \fqm^{2}: a\notin \Delta_{B}^{\bot}, a\in \Delta_{B'}^{\bot}, b\in \Delta_{A}^{\bot}, a+b \in \Delta_{A}^{\bot}\}|\\
=&q^{m-|A|}|\{a\in \fqm: a\notin \Delta_{B}^{\bot}, a\in \Delta_{B'}^{\bot}, a \in \Delta_{A}^{\bot} \}|\\
=&q^{m-|A|}(|\Delta_{A}^{\bot} \cap \Delta_{B'}^{\bot}|-|\Delta_{A}^{\bot} \cap \Delta_{B}^{\bot}|)
=q^{m-|A|}(q^{m-|A\cup B'|}-q^{m-|A\cup B|}).
\end{align*}
Moreover, according to the computation on the Lee weights of $\C_{L}$, one has 
\begin{align*} 
  N_{1}:=&|\{(a,b)\in \fqm^{2}: a\in \Delta_{B}^{\bot},b\in \Delta_{A}^{\bot}, a+b \notin \Delta_{A}^{\bot} \}| \nonumber\\
  =&|\{(a,b)\in \fqm^{2}: a\in \Delta_{B}^{\bot},b\in \Delta_{A}^{\bot}, a \notin \Delta_{A}^{\bot}\}|  \nonumber\\
  =&q^{m-|A|}(|\Delta_{B}^{\bot}|-|\Delta_{A}^{\bot}\cap \Delta_{B}^{\bot}|)\\
  =&q^{m-|A|}(q^{m-|B|}-q^{m-|A\cup B|})
\end{align*}
and similarly, by denoting $a+b=-c$, it gives
\begin{align*}
  N_{2}:=&|\{(a,b)\in \fqm^{2}: a\in \Delta_{B}^{\bot},b\notin \Delta_{A}^{\bot}, a+b \in \Delta_{A}^{\bot} \}| \nonumber\\
  =&|\{(a,c)\in \fqm^{2}: a\in \Delta_{B}^{\bot},a+c\notin \Delta_{A}^{\bot}, c \in \Delta_{A}^{\bot} \}| \nonumber\\
  =&N_{1}=q^{m-|A|}(q^{m-|B|}-q^{m-|A\cup B|}),
\end{align*}
which indicates $A_{w_{3}}=N_{1}+N_{2}=2q^{m-|A|}(q^{m-|B|}-q^{m-|A\cup B|})$. Similar to the computation on $A_{w_{3}}$, we have 
\begin{align*} 
N_{3}:=&|\{(a,b)\in \fqm^{2}: a\notin \Delta_{B}^{\bot}, a\in \Delta_{B'}^{\bot}, b\in \Delta_{A}^{\bot}, a+b \notin \Delta_{A}^{\bot} \}| \nonumber\\
=&|\{(a,b)\in \fqm^{2}:  a\notin \Delta_{B}^{\bot}, a\in \Delta_{B'}^{\bot}, b\in \Delta_{A}^{\bot}, a \notin \Delta_{A}^{\bot}\}|  \nonumber\\
=&q^{m-|A|}(|\Delta_{B'}^{\bot}\setminus \Delta_{B}^{\bot}|-|(\Delta_{B'}^{\bot}\setminus \Delta_{B}^{\bot})\cap \Delta_{A}^{\bot}|)\nonumber\\
=&q^{m-|A|}(q^{m-|B'|}-q^{m-|B|}-q^{m-|A\cup B'|}+q^{m-|A\cup B|})
\end{align*}
and by denoting $a+b=-c$ it gives
\begin{align*}
N_{4}:=&|\{(a,b)\in \fqm^{2}: a\notin \Delta_{B}^{\bot}, a\in \Delta_{B'}^{\bot}, b\notin \Delta_{A}^{\bot}, a+b \in \Delta_{A}^{\bot} \}| \nonumber\\
=&|\{(a,c)\in \fqm^{2}: a\notin \Delta_{B}^{\bot}, a\in \Delta_{B'}^{\bot},a+c\notin \Delta_{A}^{\bot}, c \in \Delta_{A}^{\bot} \}| \nonumber\\
=&N_{3}=q^{m-|A|}(q^{m-|B'|}-q^{m-|B|}-q^{m-|A\cup B'|}+q^{m-|A\cup B|}),
\end{align*}
which implies $A_{w_{4}}=N_{3}+N_{4}=2q^{m-|A|}(q^{m-|B'|}-q^{m-|B|}-q^{m-|A\cup B'|}+q^{m-|A\cup B|})$. Therefore, $A_{w_{2}}=q^{2m}-A_{w_{0}}-A_{w_{1}}-A_{w_{3}}-A_{w_{4}}=q^{2m}-q^{m-|A|}(2q^{m-|B'|}-q^{m-|A\cup B'|})$. Since each codeword in $\C_{L}$ repeats $A_{0}=q^{2m-|A|-|A\cup B|}$ times, the frequency of $w_{i}$ can be obtained from the value of $A_{w_{i}}$ by dividing $A_{0}$, where $0\leq i\leq 4$. Then the Lee weight distribution of $\C_{L}$ follows.
This completes the proof.
\end{proof}

\begin{remark}
When $q=2$ and $|B|=m$, the code $\C_{L}$ in Theorem \ref{code-1} is reduced to the linear code over $\ftwo+u\ftwo$ in \cite[Theorem 3.2]{YWXZQY}. Moreover, when $|B|=m$ and both $|A|$ and $|B'|$ divide $m$, the code $\C_{L}$ in Theorem \ref{code-1} has the same parameters and weight distribution as those of the code in \cite[Theorem 1]{HCLZ}. 
\end{remark}

\begin{remark}
Notice that the code $\C_{L}$ in Theorem \ref{code-1} is a $3$-weight code if $A\subseteq B$; it is a $2$-weight code if $A\subseteq B'$; and it is a $3$-weight code if $|A|=m$.
\end{remark}

\begin{example}
Let $q=2$, $m=6$, $A=\{1,2,3,5\}$, $B=\{1,2,3,4\}$ and $B'=\{2\}$. Magma experiments show that $\C_{L}$ is a linear code over $\ftwo+u\ftwo$ of length $224$ and size $2^{9}$, and it has the weight enumerator $1+2z^{112}+482z^{224}+26z^{240}+z^{256}$, which is consistent with our result in Theorem \ref{code-1}.
\end{example}

\subsection{The second class of linear codes  $\C_{L}$ with $L=\Delta_{A}^c+u(\Delta_{B}\setminus \Delta_{B'})$}

\begin{thm} \label{code-2}
Let $m$ be a positive integer. Let $\Delta_{A}$, $\Delta_{B}$ and $\Delta_{B'}$ be simplicial complexes of $\fqm$, where $A\subset [m]$ and $B' \subset B\subseteq [m]$. Denote $L=\Delta_{A}^c+u(\Delta_{B}\setminus \Delta_{B'})$. Then $\C_{L}$ defined by \eqref{CL} is a $9$-weight code of length $(q^{m}-q^{|A|})(q^{m}-q^{|B'|})$, size $q^{2m}$, and its Lee weight distribution is given by
\begin{center} \small  
\begin{tabular}{llll}
\hline Weight $w$ & Multiplicity $A_{w}$\\ \hline
$0$ & $ 1$ \\
$(q-1)q^{m-1}(q^{|B|}-q^{|B'|})$           & $2(q^{m-|A\cup B|}-1)$ \\
$(q-1)(q^{m-1}-q^{|A|-1})(q^{|B|}-q^{|B'|})$ & $2(q^{m-|B|}-q^{m-|A\cup B|})$ \\
$2(q-1)q^{m-1}(q^{|B|}-q^{|B'|})$ & $q^{m-|A\cup B|}(q^{m-|A|}-1)-(q^{m-|A\cup B|}-1)$ \\
$(q-1)(2q^{m-1}-q^{|A|-1})(q^{|B|}-q^{|B'|})$ & $2(q^{m-|B|}-q^{m-|A\cup B|})(q^{m-|A|}-1)$ \\
$2(q-1)(q^{m-1}-q^{|A|-1})(q^{|B|}-q^{|B'|})$ & $q^{2m}-q^{m-|A|}(2q^{m-|B'|}-q^{m-|A\cup B'|})$ \\
$(q-1)(2(q^{m}-q^{|A|})q^{|B|-1}-q^{m+|B'|-1})$ & $2(q^{m-|A\cup B'|}-q^{m-|A\cup B|})$ \\
$(q-1)(q^{m-1}-q^{|A|-1})(2q^{|B|}-q^{|B'|})$ & $2(q^{m-|B'|}-q^{m-|B|})-2(q^{m-|A\cup B'|}-q^{m-|A\cup B|})$ \\
$2(q-1)((q^{m-1}-q^{|A|-1})q^{|B|}-q^{m+|B'|-1})$ & $(q^{m-|A\cup B'|}-q^{m-|A\cup B|})(q^{m-|A|}-2)$ \\
$(q-1)(2(q^{m-1}-q^{|A|-1})(q^{|B|}-q^{|B'|})-q^{|A|+|B'|-1})$ & $2(q^{m-|B'|}-q^{m-|B|}-q^{m-|A\cup B'|}+q^{m-|A\cup B|})(q^{m-|A|}-1)$ \\
\hline
\end{tabular}
\end{center}
\end{thm}

\begin{proof}
It is clear that the length of $\C_{L}$ is $n:=(q^{m}-q^{|A|})(q^{|B|}-q^{|B'|})$. By Lemma \ref{lem-wtL}, for $a+ub \in \mathcal{R}\backslash \{0\}$, the Lee weight of the codeword $c_{a+ub}$ in $\C_{L}$ is
\begin{align*}
wt_{L}(c_{a+ub})=&2(q^{m}-q^{|A|})(q^{|B|}-q^{|B'|})-\Omega,
\end{align*}
where
\begin{align*}
\Omega=\frac{1}{q}\sum_{u\in \fq}\sum_{y \in \Delta_{B}\setminus \Delta_{B'}}\chi(u\Tr_{q}^{q^m}(ay))\sum_{x\in \Delta_{A}^c}(\chi(u\Tr_{q}^{q^m}(bx))+\chi(u\Tr_{q}^{q^m}((a+b)x))).
\end{align*}

Note that $\sum_{y \in \Delta_{B}\setminus \Delta_{B'}}\chi(u\Tr_{q}^{q^m}(ay))$ for $u\in \fq^*$ is given by \eqref{eq-thm2-Delta}. Similar to \eqref{eq-thm2-Delta}, for $u\in \fq^*$, one has 
\begin{equation} \label{eq-2-2}
  \sum_{x\in \Delta_{A}^c}\chi(u\Tr_{q}^{q^m}(bx))
  =\left\{\begin{array}{ll}
    q^{m}-q^{|A|},    &   \mbox{ if } b=0;\\
    -q^{|A|},    &   \mbox{ if } b \ne 0, b\in \Delta_{A}^{\bot};\\
    0,    &   \mbox{ if }  b\notin \Delta_{A}^{\bot}.
    \end{array} \right.
\end{equation}
To further determine the value of $\Omega$, we consider the following three cases.

Case $(1)$: $a\in \Delta_{B}^{\bot}$. Then we have 
\begin{align*}
  \Omega=\frac{2}{q}n+\frac{1}{q}(q^{|B|}-q^{|B'|})\sum_{u\in \fq^*}\sum_{x\in \Delta_{A}^c}(\chi(u\Tr_{q}^{q^m}(bx))+\chi(u\Tr_{q}^{q^m}((a+b)x))).
\end{align*}
Next we discuss the following three subcases:

Subcase $(1.1)$: $b=0$. By \eqref{eq-2-2}, we have 
\[\Omega=\left\{\begin{array}{ll}
  2(q^{m}-q^{|A|})(q^{|B|}-q^{|B'|}),    &   \mbox{ if } a+b=0;\\
  ((q+1)q^{m-1}-2q^{|A|})(q^{|B|}-q^{|B'|}),    &   \mbox{ if } a+b \ne 0, a+b\in \Delta_{A}^{\bot};\\
  (q+1)(q^{m-1}-q^{|A|-1})(q^{|B|}-q^{|B'|}),    &   \mbox{ if }  a+b\notin \Delta_{A}^{\bot},
  \end{array} \right.\]
which indicates
\begin{equation*} \label{eq-2-1}
wt_{L}(c_{a+ub})=\left\{\begin{array}{ll}
  0,    &   \mbox{ if } a+b=0;\\
  w_{1}:=(q-1)q^{m-1}(q^{|B|}-q^{|B'|}),    &   \mbox{ if } a+b \ne 0, a+b\in \Delta_{A}^{\bot};\\
  w_{2}:=(q-1)(q^{m-1}-q^{|A|-1})(q^{|B|}-q^{|B'|}),    &   \mbox{ if }  a+b\notin \Delta_{A}^{\bot}.
\end{array}\right.
\end{equation*}

Subcase $(1.2)$: $b \ne 0$, $b\in \Delta_{A}^{\bot}$. Similar to Subcase $(1.1)$, we have 
\begin{equation*} \label{eq-2-3}
wt_{L}(c_{a+ub})=\left\{\begin{array}{ll}
  w_{1}=(q-1)q^{m-1}(q^{|B|}-q^{|B'|}),    &   \mbox{ if } a+b=0;\\
  w_{3}:=2(q-1)q^{m-1}(q^{|B|}-q^{|B'|}),    &   \mbox{ if } a+b \ne 0, a+b\in \Delta_{A}^{\bot};\\
  w_{4}:=(q-1)(2q^{m-1}-q^{|A|-1})(q^{|B|}-q^{|B'|}),    &   \mbox{ if }  a+b\notin \Delta_{A}^{\bot}.
\end{array}\right.
\end{equation*}

Subcase $(1.3)$: $b\notin \Delta_{A}^{\bot}$. Similar to Subcase $(1.1)$, it gives
\begin{equation*} \label{eq-2-4}
wt_{L}(c_{a+ub})=\left\{\begin{array}{ll}
  w_{2}=(q-1)(q^{m-1}-q^{|A|-1})(q^{|B|}-q^{|B'|}),    &   \mbox{ if } a+b=0;\\
  w_{4}=(q-1)(2q^{m-1}-q^{|A|-1})(q^{|B|}-q^{|B'|}),    &   \mbox{ if } a+b \ne 0, a+b\in \Delta_{A}^{\bot};\\
  w_{5}:=2(q-1)(q^{m-1}-q^{|A|-1})(q^{|B|}-q^{|B'|}),    &   \mbox{ if }  a+b\notin \Delta_{A}^{\bot},
\end{array}\right.
\end{equation*}

Case $(2)$: $a\notin \Delta_{B}^{\bot}$, $a\in \Delta_{B'}^{\bot}$. Then we have  
\begin{align*}
  \Omega=\frac{2}{q}n-q^{|B'|-1}\sum_{u\in \fq^*}\sum_{x\in \Delta_{A}^c}(\chi(u\Tr_{q}^{q^m}(bx))+\chi(u\Tr_{q}^{q^m}((a+b)x))).
\end{align*}
Similar to Case $(1)$, we have the following results.

Subcase $(2.1)$: $b=0$. One has 
\begin{equation*}  
  wt_{L}(c_{a+ub})=\left\{\begin{array}{ll}
    w_{6}:=2(q-1)(q^{m}-q^{|A|})q^{|B|-1},    &   \mbox{ if } a+b=0;\\
    w_{7}:=(q-1)(2(q^{m}-q^{|A|})q^{|B|-1}-q^{m+|B'|-1}),    &   \mbox{ if } a+b \ne 0, a+b\in \Delta_{A}^{\bot};\\
    w_{8}:=(q-1)(q^{m-1}-q^{|A|-1})(2q^{|B|}-q^{|B'|}),    &   \mbox{ if }  a+b\notin \Delta_{A}^{\bot}.
  \end{array}\right.
\end{equation*}

Subcase $(2.2)$: $b \ne 0$, $b\in \Delta_{A}^{\bot}$. It follows that 
\begin{equation*}  
  wt_{L}(c_{a+ub})=\left\{\begin{array}{ll}
    w_{7}=(q-1)(2(q^{m}-q^{|A|})q^{|B|-1}-q^{m+|B'|-1}),    &   \mbox{ if } a+b=0;\\
    w_{9}:=2(q-1)((q^{m-1}-q^{|A|-1})q^{|B|}-q^{m+|B'|-1}),    &   \mbox{ if } a+b \ne 0, a+b\in \Delta_{A}^{\bot};\\
    w_{10}:=(q-1)(2(q^{m-1}-q^{|A|-1})(q^{|B|}-q^{|B'|})-q^{|A|+|B'|-1}),    &   \mbox{ if }  a+b\notin \Delta_{A}^{\bot}.
  \end{array}\right.
\end{equation*}

Subcase $(2.3)$: $b\notin \Delta_{A}^{\bot}$. It gives 
\begin{equation*}  
  wt_{L}(c_{a+ub})=\left\{\begin{array}{ll}
    w_{8}=(q-1)(q^{m-1}-q^{|A|-1})(2q^{|B|}-q^{|B'|}),    &   \mbox{ if } a+b=0;\\
    w_{10}=(q-1)(2(q^{m-1}-q^{|A|-1})(q^{|B|}-q^{|B'|})-q^{|A|+|B'|-1}),    &   \mbox{ if } a+b \ne 0, a+b\in \Delta_{A}^{\bot};\\
    w_{5}=2(q-1)(q^{m-1}-q^{|A|-1})(q^{|B|}-q^{|B'|}),    &   \mbox{ if }  a+b\notin \Delta_{A}^{\bot}.
  \end{array}\right.
\end{equation*}

Case $(3)$: $a\notin \Delta_{B'}^{\bot}$. Then it can be easily verified that $wt_{L}(c_{a+ub})=w_{5}=2(q-1)(q^{m-1}-q^{|A|-1})(q^{|B|}-q^{|B'|})$ in this case.

Notice that $w_{2}$ is the minimal value in the $10$ nonzero Lee weights $w_{i}$ for $1\leq i \leq 10$. It's clear that $w_{2}>0$, and thus $wt_{L}(c_{a+ub})=0$ if and only if $a=b=0$. This shows that the size of $\C_{L}$ is $q^{2m}$.

Now, we compute the Lee weight distribution of $\C_{L}$.
From the discussion on $wt_{L}(c_{a+ub})$, one has
\begin{align*}
  A_{w_{1}}=&|\{(a,b)\in \fqm^{2}: a\in \Delta_{B}^{\bot}, b=0, a+b \ne 0, a+b\in \Delta_{A}^{\bot}\}|\\
  +&|\{(a,b)\in \fqm^{2}: a\in \Delta_{B}^{\bot}, b \ne 0, b\in \Delta_{A}^{\bot}, a+b=0\}| \\
  =&2|\{(a,b)\in \fqm^{2}: a\in \Delta_{B}^{\bot}, a \ne 0, a\in \Delta_{A}^{\bot}, b=0\}|\\
  =&2(q^{m-|A\cup B|}-1),
\end{align*}
where the second equality holds since $|\{(a,b)\in \fqm^{2}: a\in \Delta_{B}^{\bot}, b \ne 0, b\in \Delta_{A}^{\bot}, a+b=0\}|=|\{(a,c)\in \fqm^{2}: a\in \Delta_{B}^{\bot}, a+c \ne 0, a+c\in \Delta_{A}^{\bot}, c=0\}|$ by denoting $c:=a+b$. 
Similarly we have 
\begin{align*}
  A_{w_{2}}=&2|\{(a,b)\in \fqm^{2}: a\in \Delta_{B}^{\bot}, b=0, a+b\notin \Delta_{A}^{\bot}\}|\\
  =&2|\{(a,b)\in \fqm^{2}: a\in \Delta_{B}^{\bot}, a\notin \Delta_{A}^{\bot}, b=0\}|\\
  =&2(q^{m-|B|}-q^{m-|A\cup B|}),
\end{align*}
and
\begin{align*}
  A_{w_{3}}=&|\{(a,b)\in \fqm^{2}: a\in \Delta_{B}^{\bot},b \ne 0, b\in \Delta_{A}^{\bot},a+b \ne 0, a+b\in \Delta_{A}^{\bot}\}| \\
  =&|\{(a,b)\in \fqm^{2}: a\in \Delta_{B}^{\bot}, a\in \Delta_{A}^{\bot}, b \ne 0, b\in \Delta_{A}^{\bot},a+b \ne 0\}|\\
  =&|\{(a,b)\in \fqm^{2}: a\in \Delta_{B}^{\bot}, a\in \Delta_{A}^{\bot}, b \ne 0, b\in \Delta_{A}^{\bot}\}|\\
  &-|\{(a,b)\in \fqm^{2}: a\in \Delta_{B}^{\bot}, a\in \Delta_{A}^{\bot}, b \ne 0, b\in \Delta_{A}^{\bot},a+b = 0\}|\\
  =&|\{(a,b)\in \fqm^{2}: a\in \Delta_{B}^{\bot}, a\in \Delta_{A}^{\bot}, b \ne 0, b\in \Delta_{A}^{\bot}\}|
  -|\{(a,b)\in \fqm^{2}: a \ne 0, a\in \Delta_{B}^{\bot}, a\in \Delta_{A}^{\bot}\}|\\
  =&q^{m-|A\cup B|}(q^{m-|A|}-1)-(q^{m-|A\cup B|}-1).
\end{align*}
Using the same techniques of computing $A_{1}$, $A_{2}$ and $A_{w_{3}}$, we have $A_{w_{4}}=2(q^{m-|B|}-q^{m-|A\cup B|})(q^{m-|A|}-1)$, $A_{w_{6}}=0$, $A_{w_{7}}=2(q^{m-|A\cup B'|}-q^{m-|A\cup B|})$, $A_{w_{8}}=2(q^{m-|B'|}-q^{m-|B|})-2(q^{m-|A\cup B'|}-q^{m-|A\cup B|})$, $A_{w_{9}}=(q^{m-|A\cup B'|}-q^{m-|A\cup B|})(q^{m-|A|}-2)$, and $A_{w_{10}}=2((q^{m-|B'|}-q^{m-|B|}) -(q^{m-|A\cup B'|}-q^{m-|A\cup B|}))(q^{m-|A|}-1)$.

With the compuation as above, we have 
$A_{w_{5}}=q^{2m}-q^{m-|A|}(2q^{m-|B'|}-q^{m-|A\cup B'|})$ due to $\sum_{1\leq i \leq 10}A_{w_{i}}+1=q^{2m}$. Then the Lee weight distribution of $\C_{L}$ follows. 
This completes the proof.
\end{proof}

\begin{remark}
When $q=2$ and $|B|=m$, the code $\C_{L}$ in Theorem \ref{code-2} is reduced to the linear code over $\ftwo+u\ftwo$ in \cite[Theorem 3.6]{YWXZQY}. Moreover, when $|B|=m$ and both $|A|$ and $|B'|$ divide $m$, the code $\C_{L}$ in Theorem \ref{code-2} has the same parameters and weight distribution as those of the code in \cite[Theorem 2]{HCLZ}. 
\end{remark}

\begin{remark}
Notice that the code $\C_{L}$ in Theorem \ref{code-2} is a $7$-weight code if $A\subseteq B$; it is a $5$-weight code if $A\subseteq B'$. Moreover, it is a $6$-weight code if $|B|=m$. 
\end{remark}

\begin{example}
Let $q=2$, $m=6$, $A=\{1,2,3,5\}$, $B=\{1,2,3,4\}$ and $B'=\{2\}$. Magma experiments show that $\C_{L}$ is a linear code over $\ftwo+u\ftwo$ of length $672$ and size $2^{12}$, and it has the weight enumerator $1+4z^{336}+2z^{448}+4z^{640}+156z^{656}+3856z^{672}+4z^{704}+52z^{720}+12z^{784}+5z^{896}$, which is consistent with our result in Theorem \ref{code-2}.
\end{example}

\subsection{The third class of linear codes $\C_{L}$ with $L=\Delta_{A}+u(\Delta_{B}\setminus \Delta_{B'})^c$}

\begin{thm}  \label{code-3}
Let $m$ be a positive integer. Let $\Delta_{A}$, $\Delta_{B}$ and $\Delta_{B'}$ be simplicial complexes of $\fqm$, where $A\subseteq [m]$ and $B'\subseteq B\subset [m]$. Denote $L=\Delta_{A}+u(\Delta_{B}\setminus \Delta_{B'})^c$. Then $\C_{L}$ defined by \eqref{CL} is a $5$-weight code of length $q^{|A|}(q^m-q^{|B|}+q^{|B'|})$, size $q^{m+|A|}$, and its Lee weight distribution is given by
\begin{center} \normalsize 
\begin{tabular}{llll}
\hline Weight $w$ & Multiplicity $A_{w}$\\ \hline
$0$ & $ 1$ \\
$2(q-1)q^{|A|-1}(q^{m}-q^{|B|}+q^{|B'|})$  & $q^{m+|A|}-2q^{m-|B'|}+q^{m-|A\cup B'|}$ \\
$2(q-1)q^{m+|A|-1}$ & $q^{m-|A\cup B|}-1$ \\
$(q-1)q^{|A|-1}(2q^{m}-q^{|B|}+q^{|B'|})$ & $2(q^{m-|B|}-q^{m-|A\cup B|})$ \\
$2(q-1)q^{|A|-1}(q^m-q^{|B|})$ & $q^{m-|A\cup B'|}-q^{m-|A\cup B|}$ \\
$(q-1)q^{|A|-1}(2q^{m}-2q^{|B|}+q^{|B'|})$ & $2(q^{m-|B'|}-q^{m-|B|}-q^{m-|A\cup B'|}+q^{m-|A\cup B|})$ \\
\hline
\end{tabular}
\end{center}
\end{thm}



\begin{proof}
The length of $\C_{L}$ is $n:=q^{|A|}(q^m-q^{|B|}+q^{|B'|})$. By Lemma \ref{lem-wtL}, for $a+ub \in \mathcal{R}\backslash \{0\}$, the Lee weight of the codeword $c_{a+ub}$ in $\C_{L}$ is
\begin{align*}
wt_{L}(c_{a+ub})=&2q^{|A|}(q^m-q^{|B|}+q^{|B'|})-\Omega,
\end{align*}
where
\begin{align*}
\Omega=\frac{1}{q}\sum_{u\in \fq}\sum_{y \in (\Delta_{B}\setminus \Delta_{B'})^c}\chi(u\Tr_{q}^{q^m}(ay))\sum_{x\in \Delta_{A}}(\chi(u\Tr_{q}^{q^m}(bx))+\chi(u\Tr_{q}^{q^m}((a+b)x))).
\end{align*}

By Lemma \ref{lem-expsum-H}, for $u\in \fq^*$, one can obtain that
\begin{align} \label{eq-thm3-B'}
\sum_{y\in (\Delta_{B}\setminus \Delta_{B'})^c}\chi(u\Tr_{q}^{q^m}(ay))=&\sum_{y\in \fqm}\chi(u\Tr_{q}^{q^m}(ay))-\sum_{y\in \Delta_{B}}\chi(u\Tr_{q}^{q^m}(ay))+\sum_{x\in \Delta_{B'}}\chi(u\Tr_{q}^{q^m}(ay))\nonumber\\
=&\left\{\begin{array}{ll}
  q^{m}-q^{|B|}+q^{|B'|},    &   \mbox{ if } a=0;\\
  -q^{|B|}+q^{|B'|},    &   \mbox{ if } a\ne 0, a\in \Delta_{B}^{\bot};\\
  q^{|B'|},    &   \mbox{ if } a\notin \Delta_{B}^{\bot}, a \in \Delta_{B'}^{\bot};\\
  0,    &   \mbox{ if } a\notin \Delta_{B'}^{\bot}.
  \end{array} \right.
\end{align}
To determine Lee weight of $\C_{L}$, we consider the following three cases.

Case $(1)$: $a=0$. Then we have 
\begin{align*}
  \Omega=&\frac{2}{q}n+\frac{2}{q}(q^{m}-q^{|B|}+q^{|B'|})\sum_{u\in \fq^*}\sum_{x\in \Delta_{A}}\chi(u\Tr_{q}^{q^m}(bx))\\
  =&\left\{\begin{array}{ll}
    2q^{|A|}(q^{m}-q^{|B|}+q^{|B'|}),    &   \mbox{ if } b \in \Delta_{A}^{\bot};\\
    2q^{|A|-1}(q^{m}-q^{|B|}+q^{|B'|}),    &   \mbox{ if } b \notin \Delta_{A}^{\bot}.
    \end{array} \right.
\end{align*}
Thus, for $a=0$, one has
\begin{equation*} 
wt_{L}(c_{a+ub})=\left\{\begin{array}{ll}
  0,    &   \mbox{ if } b \in \Delta_{A}^{\bot};\\
  w_{1}:=2(q-1)q^{|A|-1}(q^{m}-q^{|B|}+q^{|B'|}),    &   \mbox{ if }  b \notin \Delta_{A}^{\bot}.
  \end{array} \right.
\end{equation*}

Case $(2)$: $a\ne 0$, $a\in \Delta_{B}^{\bot}$. Then we have  
\begin{align*}
  \Omega=&\frac{2}{q}n-\frac{1}{q}(q^{|B|}-q^{|B'|})\sum_{u\in \fq^*}\sum_{x\in \Delta_{A}}(\chi(u\Tr_{q}^{q^m}(bx))+\chi(u\Tr_{q}^{q^m}((a+b)x)))\\
  =&\left\{\begin{array}{ll}
    2q^{|A|}(q^{m-1}-q^{|B|}+q^{|B'|}),    &   \mbox{ if } b\in \Delta_{A}^{\bot}, a+b \in \Delta_{A}^{\bot};\\
    2q^{|A|-1}(q^{m}-q^{|B|}+q^{|B'|}),    &   \mbox{ if } b\notin \Delta_{A}^{\bot}, a+b \notin \Delta_{A}^{\bot};\\
    q^{|A|-1}(2q^{m}-(q+1)q^{|B|}+(q+1)q^{|B'|}),    &   \mbox{ otherwise}.
    \end{array} \right.
\end{align*}
which indicates 
\begin{equation*} 
wt_{L}(c_{a+ub})=\left\{\begin{array}{ll}
  w_{2}:=2(q-1)q^{m+|A|-1},    &   \mbox{ if } b\in \Delta_{A}^{\bot}, a+b \in \Delta_{A}^{\bot};\\
  w_{1}=2(q-1)q^{|A|-1}(q^{m}-q^{|B|}+q^{|B'|}),    &   \mbox{ if } b\notin \Delta_{A}^{\bot}, a+b \notin \Delta_{A}^{\bot};\\
  w_{3}:=(q-1)q^{|A|-1}(2q^{m}-q^{|B|}+q^{|B'|}),    &   \mbox{ otherwise}.
  \end{array} \right.
\end{equation*}

Case $(3)$: $a\notin \Delta_{B}^{\bot}$, $a \in \Delta_{B'}^{\bot}$. Then we have  
\begin{align*}
  \Omega=&\frac{2}{q}n+q^{|B'|-1}\sum_{u\in \fq^*}\sum_{x\in \Delta_{A}}(\chi(u\Tr_{q}^{q^m}(bx))+\chi(u\Tr_{q}^{q^m}((a+b)x)))\\
  =&\left\{\begin{array}{ll}
    2q^{|A|-1}(q^{m}-q^{|B|})+2q^{|A|+|B'|},    &   \mbox{ if } b\in \Delta_{A}^{\bot}, a+b \in \Delta_{A}^{\bot};\\
    2q^{|A|-1}(q^{m}-q^{|B|}+q^{|B'|}),    &   \mbox{ if } b\notin \Delta_{A}^{\bot}, a+b \notin \Delta_{A}^{\bot};\\
    2q^{|A|-1}(q^{m}-q^{|B|})+(q+1)q^{|A|+|B'|-1},    &   \mbox{ otherwise},
    \end{array} \right.
\end{align*}
which implies 
\begin{equation*} 
wt_{L}(c_{a+ub})=\left\{\begin{array}{ll}
w_{4}:=2(q-1)q^{|A|-1}(q^m-q^{|B|}),    &   \mbox{ if } b\in \Delta_{A}^{\bot}, a+b \in \Delta_{A}^{\bot};\\
w_{1}=2(q-1)q^{|A|-1}(q^{m}-q^{|B|}+q^{|B'|}),    &   \mbox{ if } b\notin \Delta_{A}^{\bot}, a+b \notin \Delta_{A}^{\bot};\\
w_{5}:=(q-1)q^{|A|-1}(2q^{m}-2q^{|B|}+q^{|B'|}),    &   \mbox{ otherwise}.
\end{array} \right.
\end{equation*}

Case $(4)$: $a\notin \Delta_{B'}^{\bot}$. Then it can be easily verified that $wt_{L}(c_{a+ub})=w_{1}=2(q-1)q^{|A|-1}(q^{m}-q^{|B|}+q^{|B'|})$ in this case.

Notice that $0<w_{4}<w_{1},w_{2},w_{3},w_{5}$ since $|B'|<|B|<m$.
Thus $wt_{L}(c_{a+ub})=0$ if and only if $a=0$ and $b \in \Delta_{A}^{\bot}$, and it can be obtained that $A_{0}=q^{m-|A|}$. This shows that the size of $\C_{L}$ is $q^{m+|A|}$.

Now we compute the Lee weight distribution of $\C_{L}$. Similar to the computation in Theorem \ref{code-2}, it follows that 
\begin{align*}
  A_{w_{2}}=&|\{(a,b)\in \fqm^{2}: a\ne 0, a\in \Delta_{B}^{\bot},b\in \Delta_{A}^{\bot}, a+b \in \Delta_{A}^{\bot}\}| \\
  =&|\{(a,b)\in \fqm^{2}: a\ne 0, a\in \Delta_{B}^{\bot},a \in \Delta_{A}^{\bot},b\in \Delta_{A}^{\bot}\}| \\
  =&(q^{m-|A\cup B|}-1)q^{m-|A|},
\end{align*}
\begin{align*}
  A_{w_{3}}=&2|\{(a,b)\in \fqm^{2}: a\ne 0, a\in \Delta_{B}^{\bot},b\in \Delta_{A}^{\bot}, a+b \notin \Delta_{A}^{\bot}\}| \\
  =&2|\{(a,b)\in \fqm^{2}: a\in \Delta_{B}^{\bot}, a \notin \Delta_{A}^{\bot}, b\in \Delta_{A}^{\bot}\}| \\
  =&2(q^{m-|B|}-q^{m-|A\cup B|})q^{m-|A|},
\end{align*}
\begin{align*}
  A_{w_{4}}=&|\{(a,b)\in \fqm^{2}: a\notin \Delta_{B}^{\bot}, a \in \Delta_{B'}^{\bot},b\in \Delta_{A}^{\bot}, a+b \in \Delta_{A}^{\bot}\}| \\
  =&|\{(a,b)\in \fqm^{2}: a\notin \Delta_{B}^{\bot}, a \in \Delta_{B'}^{\bot}, a \in \Delta_{A}^{\bot}, b\in \Delta_{A}^{\bot}\}| \\
  =&(q^{m-|A\cup B'|}-q^{m-|A\cup B|})q^{m-|A|},
\end{align*}
and
\begin{align*}
  A_{w_{5}}=&2|\{(a,b)\in \fqm^{2}: a\notin \Delta_{B}^{\bot}, a \in \Delta_{B'}^{\bot},b\in \Delta_{A}^{\bot}, a+b \notin \Delta_{A}^{\bot}\}| \\
  =&2|\{(a,b)\in \fqm^{2}: a\notin \Delta_{B}^{\bot}, a \in \Delta_{B'}^{\bot}, a \notin \Delta_{A}^{\bot}, b\in \Delta_{A}^{\bot}\}| \\
  =&2(q^{m-|B'|}-q^{m-|B|}-q^{m-|A\cup B'|}+q^{m-|A\cup B|})q^{m-|A|}.
\end{align*}

With the computation as above, we have $A_{w_{1}}=q^{m-|A|}(q^{m+|A|}-2q^{m-|B'|}+q^{m-|A\cup B'|})$ due to $\sum_{0\leq i \leq 5}A_{w_{i}}=q^{2m}$. Then the Lee weight distribution of $\C_{L}$ can be derived completely.
This completes the proof.
\end{proof}

\begin{remark}
Note that the code $\C_{L}$ in Theorem \ref{code-3} is a $4$-weight code if $A\subseteq B$; it is a $3$-weight code if $A\subseteq B'$; and it is a $2$-weight code if $A=B'$. 
\end{remark}

\begin{example}
Let $q=2$, $m=6$, $A=\{1,2,3,5\}$, $B=\{1,2,3,4\}$ and $B'=\{2\}$. Magma experiments show that $\C_{L}$ is a linear code over $\ftwo+u\ftwo$ of length $800$ and size $2^{10}$, and it has the weight enumerator $1+2z^{768}+52z^{784}+964z^{800}+4z^{912}+z^{1024}$, which is consistent with our result in Theorem \ref{code-3}.
\end{example}

\subsection{The fourth class of linear codes  $\C_{L}$ with $L=\Delta_{A}^c+u(\Delta_{B}\setminus \Delta_{B'})^c$}

\begin{thm} \label{code-4}
Let $m$ be a positive integer. Let $\Delta_{A}$, $\Delta_{B}$ and $\Delta_{B'}$ be simplicial complexes of $\fqm$, where $A\subset [m]$ and $B' \subseteq B\subset [m]$. Denote $L=\Delta_{A}^c+u(\Delta_{B}\setminus \Delta_{B'})^c$. Then $\C_{L}$ defined by \eqref{CL} is a $10$-weight code of length $(q^{m}-q^{|A|})(q^{m}-q^{|B|}+q^{|B'|})$, size $q^{2m}$, and its Lee weight distribution is given by
\begin{center} \small  
\begin{tabular}{llll}
\hline Weight $w$ & Multiplicity $A_{w}$\\ \hline
$0$ & $ 1$ \\
$2(q-1)q^{m-1}(q^{m}-q^{|B|}+q^{|B'|})$           & $q^{m-|A|}-1$ \\
$2(q-1)(q^{m-1}-q^{|A|-1})(q^{m}-q^{|B|}+q^{|B'|})$ & $q^{2m}-q^{m-|A|}(2q^{m-|B'|}-q^{m-|A\cup B'|})$ \\
$(q-1)q^{m-1}(2q^{m}-2q^{|A|}-q^{|B|}+q^{|B'|})$ & $2(q^{m-|A\cup B|}-1)$ \\
$(q-1)(q^{m-1}-q^{|A|-1})(2q^{m}-q^{|B|}+q^{|B'|})$ & $2(q^{m-|B|}-q^{m-|A\cup B|})$ \\
$2(q-1)q^{m-1}(q^{m}-q^{|A|}-q^{|B|}+q^{|B'|})$ & $(q^{m-|A\cup B|}-1)(q^{m-|A|}-2)$ \\
$(q-1)(2q^{m-1}(q^{m}-q^{|A|})-(2q^{m-1}-q^{|A|-1})(q^{|B|}-q^{|B'|}))$ & $2(q^{m-|B|}-q^{m-|A\cup B|})(q^{m-|A|}-1)$ \\
$(q-1)(2(q^{m-1}-q^{|A|-1})(q^{m}-q^{|B|})+q^{m+|B'|-1})$ & $2(q^{m-|A\cup B'|}-q^{m-|A\cup B|})$ \\
$(q-1)(q^{m-1}-q^{|A|-1})(2q^{m}-2q^{|B|}+q^{|B'|})$ & $2(q^{m-|B'|}-q^{m-|B|}-q^{m-|A\cup B'|}+q^{m-|A\cup B|})$ \\
$(q-1)(2(q^{m-1}-q^{|A|-1})(q^{m}-q^{|B|})+2q^{m+|B'|-1})$ & $(q^{m-|A\cup B'|}-q^{m-|A\cup B|})(q^{m-|A|}-2)$ \\
$(q-1)(2(q^{m-1}-q^{|A|-1})(q^{m}-q^{|B|})+(2q^{m}-q^{|A|})q^{|B'|-1})$ & $2(q^{m-|B'|}-q^{m-|B|}-q^{m-|A\cup B'|}+q^{m-|A\cup B|})(q^{m-|A|}-1)$ \\
\hline
\end{tabular}
\end{center}
\end{thm}

\begin{proof}
The length of $\C_{L}$ is $n:=(q^{m}-q^{|A|})(q^{m}-q^{|B|}+q^{|B'|})$. By Lemma \ref{lem-wtL}, for $a+ub \in \mathcal{R}\backslash \{0\}$, the Lee weight of the codeword $c_{a+ub}$ in $\C_{L}$ is
\begin{align*}
wt_{L}(c_{a+ub})=&2(q^{m}-q^{|A|})(q^{m}-q^{|B|}+q^{|B'|})-\Omega,
\end{align*}
where
\begin{align*}
\Omega=\frac{1}{q}\sum_{u\in \fq}\sum_{y \in (\Delta_{B}\setminus \Delta_{B'})^c}\chi(u\Tr_{q}^{q^m}(ay))\sum_{x\in \Delta_{A}^c}(\chi(u\Tr_{q}^{q^m}(bx))+\chi(u\Tr_{q}^{q^m}((a+b)x))).
\end{align*}

For $u\in \fq^*$, the values of  $\sum_{y \in (\Delta_{B}\setminus \Delta_{B'})^c}\chi(u\Tr_{q}^{q^m}(ay))$ and $\sum_{x\in \Delta_{A}^c}\chi(u\Tr_{q}^{q^m}(bx))$ can be given by \eqref{eq-thm3-B'} and \eqref{eq-2-2} respectively. To further determine the Lee weights of $\C_{L}$, we consider the following four cases.

Case $(1)$: $a=0$. Then we have 
\begin{align*}
  \Omega=&\frac{2}{q}n+\frac{2}{q}(q^{m}-q^{|B|}+q^{|B'|})\sum_{u\in \fq^*}\sum_{x\in \Delta_{A}^c}\chi(u\Tr_{q}^{q^m}(bx))\\
  =&\left\{\begin{array}{ll}
    2(q^{m}-q^{|A|})(q^{m}-q^{|B|}+q^{|B'|}),    &   \mbox{ if } b=0;\\
    2(q^{m-1}-q^{|A|})(q^{m}-q^{|B|}+q^{|B'|}),    &   \mbox{ if } b \ne 0, b\in \Delta_{A}^{\bot};\\
    2(q^{m-1}-q^{|A|-1})(q^{m}-q^{|B|}+q^{|B'|}),   &   \mbox{ if }   b\notin \Delta_{A}^{\bot},
    \end{array} \right.
\end{align*}
which indicates that 
\begin{equation*} 
  wt_{L}(c_{a+ub})=\left\{\begin{array}{ll}
    0,    &   \mbox{ if } b=0;\\
    w_{1}:=2(q-1)q^{m-1}(q^{m}-q^{|B|}+q^{|B'|}),    &   \mbox{ if } b \ne 0, b\in \Delta_{A}^{\bot};\\
    w_{2}:=2(q-1)(q^{m-1}-q^{|A|-1})(q^{m}-q^{|B|}+q^{|B'|}),    &   \mbox{ if }  b\notin \Delta_{A}^{\bot}.
  \end{array}\right.
  \end{equation*}

Case $(2)$: $a\ne 0$, $a\in \Delta_{B}^{\bot}$. Then we have 
\begin{align*}
  \Omega=\frac{2}{q}n-\frac{1}{q}(q^{|B|}-q^{|B'|})\sum_{u\in \fq^*}\sum_{x\in \Delta_{A}^c}(\chi(u\Tr_{q}^{q^m}(bx))+\chi(u\Tr_{q}^{q^m}((a+b)x))).
\end{align*}
In order to compute the value of $\Omega$, we discuss the following three subcases:

Subcase $(2.1)$: $b=0$. Then it leads to $a+b\ne 0$ in this case due to $a\ne 0$. By \eqref{eq-2-2}, we have 
\[\Omega=\left\{\begin{array}{ll}
  2(q^{m-1}-q^{|A|-1})q^{m}-((q+1)q^{m-1}-2q^{|A|})(q^{|B|}-q^{|B'|}),    &   \mbox{ if } a+b \ne 0, a+b\in \Delta_{A}^{\bot};\\
  (q^{m-1}-q^{|A|-1})(2q^{m}-(q+1)q^{|B|}+(q+1)q^{|B'|}),    &   \mbox{ if }  a+b\notin \Delta_{A}^{\bot},
  \end{array} \right.\]
which indicates
\begin{equation*} 
wt_{L}(c_{a+ub})=\left\{\begin{array}{ll}
  w_{3}:=(q-1)q^{m-1}(2q^{m}-2q^{|A|}-q^{|B|}+q^{|B'|}),    &   \mbox{ if } a+b \ne 0, a+b\in \Delta_{A}^{\bot};\\
  w_{4}:=(q-1)(q^{m-1}-q^{|A|-1})(2q^{m}-q^{|B|}+q^{|B'|}),    &   \mbox{ if }  a+b\notin \Delta_{A}^{\bot}.
\end{array}\right.
\end{equation*}

Subcase $(2.2)$: $b \ne 0$, $b\in \Delta_{A}^{\bot}$. One has 
\begin{align*} 
&wt_{L}(c_{a+ub})\\
&=\left\{\begin{array}{ll}
  w_{3}=(q-1)q^{m-1}(2q^{m}-2q^{|A|}-q^{|B|}+q^{|B'|}),    &   \mbox{ if } a+b=0;\\
  w_{5}:=2(q-1)q^{m-1}(q^{m}-q^{|A|}-q^{|B|}+q^{|B'|}),    &   \mbox{ if } a+b \ne 0, a+b\in \Delta_{A}^{\bot};\\
  w_{6}:=(q-1)(2q^{m-1}(q^{m}-q^{|A|})-(2q^{m-1}-q^{|A|-1})(q^{|B|}-q^{|B'|})),    &   \mbox{ if }  a+b\notin \Delta_{A}^{\bot}.
\end{array}\right.
\end{align*}

Subcase $(2.3)$: $b\notin \Delta_{A}^{\bot}$. It can be obtained that
\begin{align*} 
&wt_{L}(c_{a+ub})\\
&=\left\{\begin{array}{ll}
  w_{4}=(q-1)(q^{m-1}-q^{|A|-1})(2q^{m}-q^{|B|}+q^{|B'|}),    &   \mbox{ if } a+b=0;\\
  w_{6}=(q-1)(2q^{m-1}(q^{m}-q^{|A|})-(2q^{m-1}-q^{|A|-1})(q^{|B|}-q^{|B'|})),    &   \mbox{ if } a+b \ne 0, a+b\in \Delta_{A}^{\bot};\\
  w_{2}=2(q-1)(q^{m-1}-q^{|A|-1})(q^{m}-q^{|B|}+q^{|B'|}),    &   \mbox{ if }  a+b\notin \Delta_{A}^{\bot}.
\end{array}\right.
\end{align*}

Case $(3)$: $a\notin \Delta_{B}^{\bot}$, $a \in \Delta_{B'}^{\bot}$. Then we have 
\begin{align*}
  \Omega=\frac{2}{q}n+q^{|B'|-1}\sum_{u\in \fq^*}\sum_{x\in \Delta_{A}^c}(\chi(u\Tr_{q}^{q^m}(bx))+\chi(u\Tr_{q}^{q^m}((a+b)x))).
\end{align*}
Then we study the following three subcases to obtain the value of $\Omega$.

Subcase $(3.1)$: $b=0$. It leads to $a+b\ne 0$ in this case due to $a\notin \Delta_{B}^{\bot}$. Then it gives 
\begin{equation*} 
wt_{L}(c_{a+ub})=\left\{\begin{array}{ll}
  w_{7}:=(q-1)(2(q^{m-1}-q^{|A|-1})(q^{m}-q^{|B|})+q^{m+|B'|-1}),    &   \mbox{ if } a+b \ne 0, a+b\in \Delta_{A}^{\bot};\\
  w_{8}:=(q-1)(q^{m-1}-q^{|A|-1})(2q^{m}-2q^{|B|}+q^{|B'|}),    &   \mbox{ if }  a+b\notin \Delta_{A}^{\bot}.
\end{array}\right.
\end{equation*}

Subcase $(3.2)$: $b \ne 0$, $b\in \Delta_{A}^{\bot}$. One has 
\begin{align*} 
&wt_{L}(c_{a+ub})\\
&=\left\{\begin{array}{ll}
  w_{7}=(q-1)(2(q^{m-1}-q^{|A|-1})(q^{m}-q^{|B|})+q^{m+|B'|-1}),    &   \mbox{ if } a+b=0;\\
  w_{9}:=(q-1)(2(q^{m-1}-q^{|A|-1})(q^{m}-q^{|B|})+2q^{m+|B'|-1}),    &   \mbox{ if } a+b \ne 0, a+b\in \Delta_{A}^{\bot};\\
  w_{10}:=(q-1)(2(q^{m-1}-q^{|A|-1})(q^{m}-q^{|B|})+(2q^{m}-q^{|A|})q^{|B'|-1}),    &   \mbox{ if }  a+b\notin \Delta_{A}^{\bot}.
\end{array}\right.
\end{align*}

Subcase $(3.3)$: $b\notin \Delta_{A}^{\bot}$. It can be derived that
\begin{align*} 
&wt_{L}(c_{a+ub})\\
&=\left\{\begin{array}{ll}
  w_{8}=(q-1)(q^{m-1}-q^{|A|-1})(2q^{m}-2q^{|B|}+q^{|B'|}),    &   \mbox{ if } a+b=0;\\
  w_{10}=(q-1)(2(q^{m-1}-q^{|A|-1})(q^{m}-q^{|B|})+(2q^{m}-q^{|A|})q^{|B'|-1}),    &   \mbox{ if } a+b \ne 0, a+b\in \Delta_{A}^{\bot};\\
  w_{2}=2(q-1)(q^{m-1}-q^{|A|-1})(q^{m}-q^{|B|}+q^{|B'|}),    &   \mbox{ if }  a+b\notin \Delta_{A}^{\bot}.
\end{array}\right.
\end{align*}

Case $(4)$: $a\notin \Delta_{B'}^{\bot}$. It gives $wt_{L}(c_{a+ub})=w_{2}=2(q-1)(q^{m-1}-q^{|A|-1})(q^{m}-q^{|B|}+q^{|B'|})$ directly in this case.

Notice that $0<w_{5}<w_{1},w_{2},w_{3},w_{4},w_{6}$ and $0<w_{8}<w_{7},w_{9},w_{10}$. By a straightforward computation, one has $w_{5}\leq w_{8}$ if and only if $m-|A|\leq |B|-|B'|$. Therefore $wt_{L}(c_{a+ub})=0$ if and only if $a=b=0$, which implies $A_{0}=1$. This shows that the size of $\C_{L}$ is $q^{2m}$.

Now we compute the Lee weight distribution of $\C_{L}$. According to the discussion on the Lee weights $wt_{L}(c_{a+ub})$ of $\C_{L}$, we can get that $A_{w_{1}}=q^{m-|A|}-1$, $A_{w_{3}}=2(q^{m-|A\cup B|}-1)$, $A_{w_{4}}=2(q^{m-|B|}-q^{m-|A\cup B|})$, $A_{w_{5}}=(q^{m-|A\cup B|}-1)(q^{m-|A|}-2)$, $A_{w_{6}}=2(q^{m-|B|}-q^{m-|A\cup B|}) (q^{m-|A|}-1)$, $A_{w_{7}}=2(q^{m-|A\cup B'|}-q^{m-|A\cup B|})$, $A_{w_{8}}=2(q^{m-|B'|}-q^{m-|B|}-q^{m-|A\cup B'|}+q^{m-|A\cup B|})$, $A_{w_{9}}=(q^{m-|A\cup B'|}-q^{m-|A\cup B|})(q^{m-|A|}-2)$ and $A_{w_{10}}=2(q^{m-|B'|}-q^{m-|B|}-q^{m-|A\cup B'|}+q^{m-|A\cup B|})(q^{m-|A|}-1)$. Here we omit the detailed computations since this can be derived similarly to the computation of $A_{w_{i}}$'s in Theorem \ref{code-2}. Furthermore, it gives $A_{w_{2}}=q^{2m}-q^{m-|A|}(2q^{m-|B'|}-q^{m-|A\cup B'|})$ by $\sum_{0\leq i \leq 10}A_{w_{i}}=q^{2m}$. Then the Lee weight distribution of $\C_{L}$ is completely determined. This completes the proof.

\end{proof}

\begin{remark}
Note that the code $\C_{L}$ in Theorem \ref{code-4} is an $8$-weight code if $A\subseteq B$; it is a $6$-weight code if $A\subseteq B'$; and it is an $8$-weight code if $|A\cup B'|=|A\cup B|$.
\end{remark}

\begin{example}
Let $q=2$, $m=6$, $A=\{1,2,3,5\}$, $B=\{1,2,3,4\}$ and $B'=\{2\}$. Magma experiments show that $\C_{L}$ is a linear code over $\ftwo+u\ftwo$ of length $2400$ and size $2^{12}$, and it has the weight enumerator $1+2z^{2176}+12z^{2288}+52z^{2352}+4z^{2368}+3856z^{2400}+156z^{2416}+4z^{2432}+2z^{2624}+4z^{2736}+3z^{3200}$, which is consistent with our result in Theorem \ref{code-4}.
\end{example}

\section{Optimal codes over $\fq$ and examples} \label{sect-4}
It's known that the Gray map $\phi $ introduced in Section \ref{sect-2} is an isometry from $(R^{n},d_{L})$ and $(\fq^{2n},d_{H})$, which is distance-preserving and weight-preserving. In this section, we will investigate the Gray images $\phi(\C_{L})$ of the codes $\C_{L}$ over $R=\fq+u\fq$ ($u^2=0$) constructed in Section \ref{sect-3}. By using the Griesmer bound, several calsses of optimal few-weight linear codes over $\fq$ and some examples will be presented.

\begin{thm} \label{thm1-cor1}
Let $\C_{L}$ be defined as in Theorem \ref{code-1} with $|A\cup B|=|B|$. Assume that $|A|+|B'|> 0$. Then the Gray image $\phi(\C_{L})$ is a $[2q^{|A|}(q^{|B|}-q^{|B'|}),|A|+|B|,2(q-1)q^{|A|-1}(q^{|B|}-q^{|B'|})]$ linear code over $\fq$ with the weight distribution 
\begin{center} 
\begin{tabular}{llll}
\hline Weight $w$ & Multiplicity $A_{w}$\\ \hline
$0$ & $1$ \\
$2(q-1)q^{|A|+|B|-1}$               & $q^{|B|-|A\cup B'|}-1$ \\
$2(q-1)q^{|A|-1}(q^{|B|}-q^{|B'|})$ & $q^{|A|+|B|}-2q^{|B|-|B'|}+q^{|B|-|A\cup B'|}$ \\ 
$(q-1)q^{|A|-1}(2q^{|B|}-q^{|B'|})$ & $2(q^{|B|-|B'|}-q^{|B|-|A\cup B'|})$ \\
\hline
\end{tabular}
\end{center}
Moreover, the code $\phi(\C_{L})$ is a near Griesmer code and it is distance-optimal.
\end{thm}

\begin{proof}
According to Theorem \ref{code-1}, the weight distribution of $\phi(\C_{L})$ can be given as in Theorem \ref{thm1-cor1} for the case $|A\cup B|=|B|$. It's clear that the minimum distance of $\phi(\C_{L})$ is $d=2(q-1)q^{|A|-1}(q^{|B|}-q^{|B'|})$, and $A_{d}=q^{2m-|A|-|A\cup B'|}(q^{|A|+|A\cup B'|}-2q^{|A\cup B'|-|B'|}+1)>0$ due to $(|A|+|A\cup B'|)-(|A\cup B'|-|B'|)=|A|+|B'|>0$. Then $\phi(\C_{L})$ has parameters $[2q^{|A|}(q^{|B|}-q^{|B'|}),|A|+|B|,d:=2(q-1)q^{|A|-1}(q^{|B|}-q^{|B'|})]$. 

By the Griesmer bound, we have
\begin{align} \label{eq-thm1-gmd}
g(|A|+|B|,d)=&\sum_{i=0}^{|A|+|B|-1} \lceil \frac{2(q-1)q^{|A|-1}(q^{|B|}-q^{|B'|})}{q^i}\rceil \nonumber\\
=&\sum_{i=0}^{|A|+|B'|-1}2(q-1)(q^{|A|+|B|-i-1}-q^{|A|+|B'|-i-1}) \nonumber\\
&+\sum_{i=|A|+|B'|}^{|A|+|B|-1}(2(q-1)q^{|A|+|B|-i-1}+\lceil -2(q-1)q^{|A|+|B'|-i-1}\rceil)  \nonumber\\
=&2q^{|A|}(q^{|B|}-q^{|B'|})-1,
\end{align}
where the last equality holds since $\lceil -2(q-1)q^{|A|+|B'|-i-1}\rceil=-1$ if $i=|A|+|B'|$ and it is equal to $0$ if $|A|+|B'|<i\leq |A|+|B|-1$.
Thus $\phi(\C_{L})$ is a near Griesmer code by \eqref{eq-thm1-gmd}. Moreover, we similarly have
\begin{align} \label{eq-thm1-gmd+1}
  g(|A|+|B|,d+1)=&\sum_{i=0}^{|A|+|B|-1} \lceil \frac{2(q-1)q^{|A|-1}(q^{|B|}-q^{|B'|})+1}{q^i}\rceil \nonumber\\
  =&2q^{|A|}(q^{|B|}-q^{|B'|})+|A|+|B'|+\lfloor \frac{2}{q}\rfloor-1,
\end{align}
which implies $\phi(\C_{L})$ is distance-optimal due to $2q^{|A|}(q^{|B|}-q^{|B'|})<g(|A|+|B|,d+1)$. This completes the proof.
\end{proof}

\begin{remark}
When $q=2$ and $|B|=m$, the optimal codes $\phi(\C_{L})$ in Theorem \ref{thm1-cor1} are reduced to the binary optimal linear codes in \cite[Theorem 4.1]{YWXZQY}. 
\end{remark}

\begin{example}
Let $q=3$, $m=4$, $A=\{1\}$, $B=\{1,2,3\}$ and $B'=\{2\}$. Magma experiments show that $\phi(\C_{L})$ is a $[144,4,96]$ linear code over ${\mathbb F}_{3}$ with the weight enumerator $1+66z^{96}+12z^{102}+2z^{108}$, which is consistent with our result in Theorem \ref{thm1-cor1}. This code is a near Griesmer code by the Griesmer bound and is optimal due to \cite{GMB}. 
\end{example}

By characterizing the optimality of the codes in Theorem \ref{code-2}, we get a class of six-weight linear codes in the following, which can produce many optimal linear codes over $\fq$.
\begin{thm} \label{thm2-cor1}
Let $\C_{L}$ be defined as in Theorem \ref{code-2} with $|B|=m$. Assume that $|A\cup B'|<m$ and $q^{m-|A|}>2$. Then the Gray image $\phi(\C_{L})$ is a $[2(q^{m}-q^{|A|})(q^{m}-q^{|B'|}),2m,2(q-1)(q^{2m-1}-q^{m+|A|-1}-q^{m+|B'|-1})]$ linear codes with the weight distribution 
\begin{center}\small 
\begin{tabular}{llll} 
\hline Weight $w$ & Multiplicity $A_{w}$\\ \hline
$0$ & $ 1$ \\
$2(q-1)q^{m-1}(q^{m}-q^{|B'|})$ & $q^{m-|A|}-1$ \\
$2(q-1)(q^{m-1}-q^{|A|-1})(q^{m}-q^{|B'|})$ & $q^{2m}-q^{m-|A|}(2q^{m-|B'|}-q^{m-|A\cup B'|})$ \\
$(q-1)(2q^{2m-1}-2q^{m+|A|-1}-q^{m+|B'|-1})$ & $2(q^{m-|A\cup B'|}-1)$ \\
$(q-1)(q^{m-1}-q^{|A|-1})(2q^{m}-q^{|B'|})$ & $2(q^{m-|B'|}-q^{m-|A\cup B'|})$ \\
$2(q-1)(q^{2m-1}-q^{m+|A|-1}-q^{m+|B'|-1})$ & $(q^{m-|A\cup B'|}-1)(q^{m-|A|}-2)$ \\
$(q-1)(2(q^{m-1}-q^{|A|-1})(q^{m}-q^{|B'|})-q^{|A|+|B'|-1})$ & $2(q^{m-|B'|}-q^{m-|A\cup B'|})(q^{m-|A|}-1)$ \\
\hline
\end{tabular}
\end{center}
Moreover, the code $\phi(\C_{L})$ is distance-optimal if $2q^{|A|+|B'|}<m+\min\{|A|,|B'|\}+\delta$, where
\begin{equation} \label{cor1-delta}
\delta=\left\{\begin{array}{ll}
1, & \mbox{if $q=2$},  \\
-1,  & \mbox{if $|A|=|B'|$ and $q>4$},  \\
0,  & \mbox{otherwise}.
\end{array}\right.
\end{equation}
\end{thm}
  
\begin{proof}
According to Theorem \ref{code-2}, $\phi(\C_{L})$ is reduced to a $6$-weight code for the case $|B|=m$, and its weight distribution follows as in Theorem \ref{thm2-cor1}. It can be verified that the minimum distance of $\phi(\C_{L})$ is $d=2(q-1)(q^{2m-1}-q^{m+|A|-1}-q^{m+|B'|-1})$, and $A_{d}=(q^{m-|A\cup B'|}-1)(q^{m-|A|}-2)>0$ due to $|A\cup B'|<m$ and $q^{m-|A|}>2$. Then $\phi(\C_{L})$ has parameters $[2(q^{m}-q^{|A|})(q^{m}-q^{|B'|}),2m,2(q-1)(q^{2m-1}-q^{m+|A|-1}-q^{m+|B'|-1})]$. 

With detailed computation by using the Griesmer bound, it gives
\begin{equation} \label{eq-gmd-cor1-1}
g(2m,d)=\left\{\begin{array}{ll}
2(q^{2m}-q^{m+|A|}-q^{m+|B'|})-1, & \mbox{if $|A|=|B'|$ and $q\not=3$};  \\
2(q^{2m}-q^{m+|A|}-q^{m+|B'|}),  & \mbox{otherwise},
\end{array}\right.
\end{equation}
and
\[g(2m,d+1)=2(q^{2m}-q^{m+|A|}-q^{m+|B'|})+m+\min\{|A|,|B'|\}+\delta,\]
where $\delta $ is defined as in \eqref{cor1-delta}.
Then $\phi(\C_{L})$ is distance-optimal if $2q^{|A|+|B'|}<m+\min\{|A|,|B'|\}+\delta$. This completes the proof.
\end{proof}
\begin{remark}
It should be noted that for the case $|A|=|B'|=0$ the code $\phi(\C_{L})$ in Corollary \ref{thm2-cor1} is close to the Griesmer bound since $n-g(2m,d)=2 \mbox{ or } 3$ by \eqref{eq-gmd-cor1-1}. Morover, the given condition in Theorem \ref{thm2-cor1} for the code $\phi(\C_{L})$ to be distance-optimal can be easily satisfied if $|A|+|B'|$ is small enough and $m$ is large enough. Thus many distance-optimal linear codes over $\fq$ can be derived from this construction.
\end{remark}


\begin{example}
Let $q=2$, $m=4$, $A=\{2\}$, $B=\{1,2,3,4\}$ and $B'=\emptyset$. Magma experiments show that $\phi(\C_{L})$ is a $[420,8,208]$ linear code over ${\mathbb F}_{2}$ with the weight enumerator $1+42z^{208}+112z^{209}+64z^{210}+14z^{216}+16z^{217}+7z^{240}$, which is consistent with our result in Theorem \ref{thm2-cor1}. This code is distance-optimal by the Griesmer bound. 
\end{example}

For a special case of Theorem \ref{code-2}, we can get a family of two-weight optimal codes as follows.
\begin{thm} \label{thm2-cor3}
Let $\C_{L}$ be defined as in Theorem \ref{code-2} with $q=2$, $|B|=m$, $A=B'$ and $|A|=|B'|=m-1$. Then the Gray image $\phi(\C_{L})$ is a $[2^{2m-1},2m,2^{2m-2}]$ linear code over $\ftwo$ with the weight enumerator $1+z^{2^{2m-1}}+(2^{2m}-2)z^{2^{2m-2}}$. This code $\phi(\C_{L})$ is a Griesmer code.
\end{thm}

\begin{proof}
This result can be proved directly by Theorem \ref{code-2} and thus we
omit the proof here.
\end{proof}
\begin{remark}
The optimal code $\phi(\C_{L})$ in Theorem \ref{thm2-cor3} are the same as the binary optimal linear codes in \cite[Theorem 4.3]{YWXZQY}. 
\end{remark}

\begin{example}
Let $q=2$, $m=4$, $A=\{1,2,3\}$, $B=\{1,2,3,4\}$ and $B'=\{1,2,3\}$. Magma experiments show that $\phi(\C_{L})$ is a $[128,8,64]$ linear code over ${\mathbb F}_{2}$ with the weight enumerator $1+254z^{64}+z^{128}$, which is consistent with our result in Theorem \ref{thm2-cor3}. This code is a Griesmer code by the Griesmer bound and it is optimal due to \cite{GMB}. 
\end{example}

In the following theorem, we investigate the optimality of the codes in Theorem \ref{code-3}, and thus many distance-optimal linear codes over $\fq$ can be derived from our contruction.
\begin{thm} \label{thm3-cor1}
Let $\C_{L}$ be defined as in Theorem \ref{code-3}. Assume that $|A\cup B'| \ne |A\cup B|$. Then the Gray image $\phi(\C_{L})$ is a $5$-weight $[2q^{|A|}(q^m-q^{|B|}+q^{|B'|}),m+|A|,2(q-1)q^{|A|-1}(q^m-q^{|B|})]$ linear code over $\fq$ with the weight distribution\end{thm}
\begin{center} \normalsize 
\begin{tabular}{llll}
\hline Weight $w$ & Multiplicity $A_{w}$\\ \hline
$0$ & $ 1$ \\
$2(q-1)q^{|A|-1}(q^{m}-q^{|B|}+q^{|B'|})$  & $q^{m+|A|}-2q^{m-|B'|}+q^{m-|A\cup B'|}$ \\
$2(q-1)q^{m+|A|-1}$ & $q^{m-|A\cup B|}-1$ \\
$(q-1)q^{|A|-1}(2q^{m}-q^{|B|}+q^{|B'|})$ & $2(q^{m-|B|}-q^{m-|A\cup B|})$ \\
$2(q-1)q^{|A|-1}(q^m-q^{|B|})$ & $q^{m-|A\cup B'|}-q^{m-|A\cup B|}$ \\
$(q-1)q^{|A|-1}(2q^{m}-2q^{|B|}+q^{|B'|})$ & $2(q^{m-|B'|}-q^{m-|B|}-q^{m-|A\cup B'|}+q^{m-|A\cup B|})$ \\
\hline
\end{tabular}
\end{center}
The code $\phi(\C_{L})$ is distance-optimal if $2q^{|A|+|B'|}<|A|+|B|+\lfloor \frac{2}{q}\rfloor-1$.

\begin{proof}
According to proof of Theorem \ref{code-3}, the minimum distance of $\phi(\C_{L})$ is $d=2(q-1)q^{|A|-1}(q^m-q^{|B|})$, and $A_{d}=q^{m-|A\cup B'|}-q^{m-|A\cup B|}>0$ due to $|A\cup B'| \ne |A\cup B|$. Thus the parameters are determied completely. Moreover, the weight distribution of 
$\phi(\C_{L})$ is the same as the Lee weight distribution of $\C_{L}$.

Now we investigate the optimality of $\phi(\C_{L})$. Due to \eqref{eq-thm1-gmd} and \eqref{eq-thm1-gmd+1}, we have 
\[g(m+|A|,d)=\sum_{i=0}^{m+|A|-1} \lceil \frac{2(q-1)q^{|A|-1}(q^m-q^{|B|})}{q^i}\rceil=2q^{|A|}(q^{m}-q^{|B|})-1\]
and 
\[g(m+|A|,d+1)=\sum_{i=0}^{m+|A|-1} \lceil \frac{2(q-1)q^{|A|-1}(q^m-q^{|B|})+1}{q^i}\rceil=2q^{|A|}(q^{m}-q^{|B|})+|A|+|B|+\lfloor \frac{2}{q}\rfloor-1.\]
Therefore the code $\phi(\C_{L})$ is distance-optimal if $2q^{|A|+|B'|}<|A|+|B|+\lfloor \frac{2}{q}\rfloor-1$.
This completes the proof.
\end{proof}

\begin{remark}
The code $\phi(\C_{L})$ in Theorem \ref{thm3-cor1} is close to the Griesmer bound if $|A|=|B'|=0$ since $n-g(m+|A|,d)=3$ in this case. Morover, observe that the given condition in Theorem \ref{thm3-cor1} for the code $\phi(\C_{L})$ to be distance-optimal can be easily satisfied if $|B|-|B'|$ is large enough.
\end{remark}
  
\begin{example}
Let $q=2$, $m=6$, $A=\{4\}$, $B=\{1,2,3,5\}$ and $B'=\emptyset$. Magma experiments show that $\phi(\C_{L})$ is a $[196,7,96]$ linear code over ${\mathbb F}_{2}$ with the weight enumerator $1+30z^{96}+60z^{97}+32z^{98}+4z^{113}+z^{128}$, which is consistent with our result in Theorem \ref{thm3-cor1}. This code is distance-optimal by the Griesmer bound and it is optimal due to \cite{GMB}. 
\end{example}

By a in-depth study on a special case of Theorem \ref{code-4}, we can obtain a class of two-weight optimal linear codes in the following.
\begin{thm} \label{thm4-cor1}
Let $\C_{L}$ be defined as in Theorem \ref{code-4} with $B=B'$ (i.e., $(\Delta_{B}\setminus \Delta_{B'})^c=\fqm$). Then the Gray image $\phi(\C_{L})$ has parameters $[2(q^{2m}-q^{m+|A|}),2m,2(q-1)(q^{2m-1}-q^{m+|A|-1})]$ with the weight enumerator $1+(q^{m-|A|}-1)z^{2(q-1)q^{2m-1}}+(q^{2m}-q^{m-|A|})z^{2(q-1)(q^{m-1}-q^{|A|-1})q^{m}}$. This code $\phi(\C_{L})$ is a near Griesmer code and it is distance-optimal.
\end{thm}

\begin{proof}
By Theorem \ref{code-4}, it can be verified directly that $\phi(\C_{L})$ is reduced to a two-weight code with the weight enumerator $1+(q^{m-|A|}-1)z^{2(q-1)q^{2m-1}}+(q^{2m}-q^{m-|A|})z^{2(q-1)(q^{m-1}-q^{|A|-1})q^{m}}$ for the case $B=B'$. Certainly, the parameters of $\phi(\C_{L})$ can be given by $[2(q^{2m}-q^{m+|A|}),2m,2(q-1)(q^{2m-1}-q^{m+|A|-1})]$. 

By the Griesmer bound, one has 
\[g(2m,d)=\sum_{i=0}^{2m-1} \lceil \frac{2(q-1)(q^{2m-1}-q^{m+|A|-1})}{q^i}\rceil=2(q^{2m}-q^{m+|A|})-1\]
and
\[g(2m,d+1)=\sum_{i=0}^{2m-1} \lceil \frac{2(q-1)(q^{2m-1}-q^{m+|A|-1})+1}{q^i}\rceil=2(q^{2m}-q^{m+|A|})+m+|A|+\lfloor \frac{2}{q}\rfloor-1.\]
Therefore the code $\phi(\C_{L})$ is a near Griesmer code and it is  distance-optimal.
This completes the proof.
\end{proof}

\begin{remark}
It should be noted that for the case $|A|=|B|=m$ the codes in Theorem \ref{thm1-cor1} have parameters $[2(q^{2m}-q^{m+|B'|}),2m,2(q-1)(q^{2m-1}-q^{m+|B'|-1})]$, which is the same as the parameters of the codes in Theorem \ref{thm4-cor1}. However, they are inequivalent since they have different weight distributions. The weight enumerator of the code in Theorem \ref{thm1-cor1} for the case $|A|=|B|=m$ can be given by $1+(q^{2m}-2q^{m-|B'|}+1)z^{2(q-1)q^{m-1}(q^{m}-q^{|B'|})}+2(q^{m-|B'|}-1)z^{(q-1)q^{m-1}(2q^{m}-q^{|B'|})}$, which is different from that of Theorem \ref{thm4-cor1}.
\end{remark}

\begin{example}
Let $q=2$, $m=4$, $A=\{1,2\}$ and $B=B'$. Magma experiments show that $\phi(\C_{L})$ is a $[384,8,192]$ linear code over ${\mathbb F}_{2}$ with the weight enumerator $1+252z^{192}+3z^{256}$, which is consistent with our result in Theorem \ref{thm4-cor1}. This code is distance-optimal by the Griesmer bound. 
\end{example}

\section{Conclusions} \label{sect-conclusion}
In this paper, we constructed four families of linear codes over $\fq+u\fq$, $u^2=0$ with defining sets associated with simplicial complexes of $\fq^m$. This extends the results of \cite{YWXZQY} from $q=2$ and $B=[m]$ to general $q$ and $B\subseteq [m]$. By computing certain exponential sums on simplicial complexes, we completely determined the parameters and Lee weight distributions of these four families of codes, and many linear codes with few Lee weights can be produced. Moreover, via the Gray map, we obtained several infinite families of optimal linear codes over $\fq$ by using the Griesmer bound, which include the near Griesmer codes and distance-optimal codes.

\section{Acknowledgements}
This work was supported by the Major Program(JD) of Hubei Province (No. 2023BAA027), the National Natural Science Foundation of China (No. 62072162), the Natural Science Foundation of Hubei Province of China (No. 2021CFA079), the Knowledge Innovation Program of Wuhan-Basic Research (No. 2022010801010319) and the innovation group project of the natural science foundation of Hubei Province of China (No. 2023AFA021).



\begin{thebibliography}{99}

\bibitem{ARJD} R. J. Anderson, C. Ding, T. Hellsesth, T. Klove, How to build robust shared control systems, Des. Codes Cryptogr. 15(2) (1998), pp. 111-123.


\bibitem{CAGJ} A. R. Calderbank, J. Goethals, Three-weight codes and association schemes, Philips J. Res. 39(4-5) (1984), pp. 143-152.

\bibitem{CCDY} C. Carlet, C. Ding, J. Yuan, Linear codes from perfect nonlinear mappings and their secret sharing schemes, IEEE Trans. Inf. Theory 51(6) (2005), pp. 2089-2102.


\bibitem{CHJY} S. Chang, J.Y. Hyun, Linear codes from simplicial complexes, Des. Codes Cryptogr. 86 (2018), pp. 2167-2181.

\bibitem{DING} C. Ding, Linear codes from some $2$-designs, IEEE Trans. Inf. Theory 61(6) (2015), pp. 3265-3275.

\bibitem{DCHT} C. Ding, T. Helleseth, T. Kl{\o}ve, X. Wang, A generic construction of cartesian authentication codes, IEEE Trans. Inf. Theory 53(6) (2007), pp. 2229-2235.

\bibitem{DN} C. Ding, H. Niederreiter, Cyclotomic linear codes of order $3$, IEEE Trans. Inf. Theory 53(6) (2007), pp. 2274-2277.

\bibitem{DW} C. Ding, X. Wang, A coding theory construction of new systematic authentication codes, Theor. Comput. Sci. 330(1) (2005), pp. 81-99.

\bibitem{GMB}M. Grassl, Bounds on the minimum distance of linear codes, Online available at http://www.codetables.de, Accessed on 2024-06-21

\bibitem{JHG} J.H. Griesmer, A bound for error correcting codes, IBM J. Res. Dev. 4 (1960), pp. 532-542.

\bibitem{HZHZ} Z. Heng, Projective linear codes from some almost difference sets, IEEE Trans. Inf. Theory 69(2) (2023), pp. 978-994.

\bibitem{HDWZ} Z. Heng, C. Ding, W. Wang, Optimal binary linear codes from maximal arcs, IEEE Trans. Inf. Theory 66(9) (2020), pp. 5387-5394.

\bibitem{HCLZ} Z. Hu, B. Chen, N. Li, X. Zeng, Two classes of optimal few-weight codes over $\fq+u\fq$, In: S. Mesnager, Z. Zhou (eds) Arithmetic of Finite Fields. WAIFI 2022. LNCS, Springer, Cham, 13638 (2023): 208-220.

\bibitem{HXLZWT} Z. Hu, Y. Xu, N. Li, X. Zeng, L. Wang, X. Tang, New constructions of optimal linear codes from simplicial complexes, IEEE Trans. Inf. Theory 70(3) (2024), pp. 1823-1835. 

\bibitem{HWPV} W. Huffman, V. Pless, Fundamentals of error-correcting codes, Cambridge University Press (1997).

\bibitem{JYHLL} J. Y. Hyun, J. Lee, Y. Lee, Infinite families of optimal linear codes constructed from simplicial complexes, IEEE Trans. Inf. Theory 66(11) (2020), pp. 6762-6773.

\bibitem{LXSM} X. Li, M. Shi, A new family of optimal binary few-weight codes from simplicial complexes, IEEE Communications Letters 25(4) (2021), pp. 1048-1051.

\bibitem{LXYQ} X. Li, Q. Yue, D. Tang, A family of linear codes from constant dimension subspace codes, Des. Codes Cryptogr. 90 (2022), pp. 1-15.

\bibitem{HLYM} H. Liu, Y. Maouche, Two or few-weight trace codes over ${\mathbb{F}_{q}}+u{\mathbb{F}_{q}}$, IEEE Trans. Inf. Theory 65(5) (2019), pp. 2696-2703.

\bibitem{MQRT} S. Mesnager, Y. Qi, H. Ru, C. Tang, Minimal linear codes from characteristic functions, IEEE Trans. Inf. Theory 66(9) (2020), pp. 5404-5413.

\bibitem{MAS} S. Mesnager, A. S{\i}nak, Several classes of minimal linear codes with few weights from weakly regular plateaued functions, IEEE Trans. Inf. Theory, 66(4) (2020), pp. 2296-2310.

\bibitem{MSYGPS} M. Shi, Y. Guan, P. Sol\'{e}, Two new families of two-weight codes, IEEE Trans. Inf. Theory 63(10) (2017), pp. 6240-6246.

\bibitem{SMLX} M. Shi, X. Li, Two classes of optimal $p$-ary few-weight codes from down-sets, Discret. Appl. Math. 290 (2021), pp. 60-67.

\bibitem{MSYLPS} M. Shi, Y. Liu, P. Sol\'{e}, Optimal two-weight codes from trace codes over $\mathbb {F}_2+u\mathbb {F}_2$, IEEE Commun. Lett. 20(12) (2016), pp. 2346-2349.

\bibitem{SMWRLY} M. Shi, R. Wu, Y. Liu, P. Sol\'{e}, Two and three weight codes over $\fp+u\fp$, Cryptogr. Commun. 9 (2017), pp. 637-646.

\bibitem{GSJS} G. Solomon, J.J. Stiffer, Algebraically punctured cyclic codes, Inform. and Control 8 (1965), pp. 170-179.

\bibitem{WYHJ} Y. Wu, J.Y. Hyun, Few-weight codes over $\fp+u\fp$ associated with down sets and their distance optimal Gray image, Discret. Appl. Math. 283 (2020), pp. 315-322.

\bibitem{WLZX} Y. Wu, C. Li, L. Zhang, F. Xiao, Quaternary codes and their binary images, IEEE Trans. Inf. Theory 70(7) (2024), pp. 4759-4768. 

\bibitem{YWXZQY} Y. Wu, X. Zhu, Q. Yue, Optimal few-weight codes from simplicial complexes, IEEE Trans. Inf. Theory 66(6) (2020), pp. 3657-3663.

\bibitem{YLFL} M. Yang, J. Li, K. Feng, D. Lin, Generalized Hamming weights of irreducible cyclic codes, IEEE Trans. Inf. Theory 61(9) (2015), pp. 4905-4913.

\bibitem{ZZNL} Z. Zhou, N. Li, C. Fan, T. Helleseth, Linear codes with two or three weights from quadratic Bent functions, Des. Codes Cryptogr. 81 (2016), pp. 283-295.

\end{thebibliography}
\end{document}